\documentclass[sigconf,nonacm]{acmart}

\setcopyright{none} 

\usepackage{amsmath} 
\usepackage{lipsum,etoolbox}
\usepackage{amsfonts} 
\usepackage{mathtools} 
\usepackage{url}
\usepackage[colorinlistoftodos]{todonotes}
\usepackage{epstopdf}
\usepackage{array}
\usepackage{subcaption}
\usepackage{soul}
\usepackage{caption}
\usepackage{graphicx}
\usepackage{tabularx,environ}
\usepackage{xspace}
\usepackage[linesnumbered,ruled]{algorithm2e}
\SetInd{0em}{0.7em}
\usepackage{multirow}

\usepackage{hyperref}

\usepackage{comment}

\usepackage{makecell}

\graphicspath{{ ../eps/}}

\newcounter{todocounter}

\SetAlFnt{\scriptsize} 
\SetKwData{Left}{left}
\SetKwData{This}{this}
\SetKwData{Up}{up}
\SetKwFunction{Union}{Union}
\SetKwFunction{FindCompress}{FindCompress}
\SetKwInput{Input}{input}
\SetKwInOut{Output}{output}
\SetKwInOut{Parameter}{parameter}

\SetKwBlock{KwFunction}{function}{}
\SetKwBlock{KwOn}{on}{}
\SetKwBlock{KwPeriodically}{periodically}{}
\SetKwBlock{KwState}{state}{}
\SetKwBlock{KwAtomic}{atomically}{}
\SetKwBlock{KwConcurrently}{concurrently}{}
\SetCommentSty{textup} 
\SetKwComment{KwIttayComment}{(}{)} 
\SetKw{KwAnd}{and}
\SetKw{KwOr}{or}
\SetKw{KwSetTimer}{setTimer}
\SetKwBlock{KwExpTimer}{on timer expiration}{}
\SetKwFor{ForAll}{forall}{}

\SetKwFunction{FAssert}{assert}

\SetKwFunction{FAllocateBudgt}{\funcNameAllocateBudgt}
\SetKwFunction{FGenerateBlock}{\funcNameGenerateBlock}
\SetKwFunction{FPublish}{\funcNamePublish}

\SetKwProg{Fn}{Function}{:}{}

\SetKwRepeat{Do}{do}{while}

\newtheorem*{myNote}{Note}
\newtheorem*{myExample}{Example}

\DeclarePairedDelimiter\floor{\lfloor}{\rfloor}
\DeclareMathOperator{\E}{\mathbb{E}}
\DeclareMathOperator{\N}{\mathbb{N}}

\newcommand{\VERTICALSUBFIGSCALE}{0.3}
\newcommand{\VERTICALFIGSCALEINPUT}{0.3}

\newcommand{\epochLength}{\ensuremath{\ell}}
\newcommand{\posRatio}{\ensuremath{\rho}}
\newcommand{\oneMinusPosRatioBrackets}{\ensuremath{\left(1-\posRatio\right)}}
\newcommand{\factor}{\ensuremath{F}}

\newcommand{\tokenCount}{\ensuremath{t}}

\newcommand{\rewardPerBlock}{\ensuremath{r}\xspace}

\newcommand{\rewardPerBlockWithEpoch}{\ensuremath{r_{\epochIndex}}\xspace}

\newcommand{\globalStorage}{\ensuremath{S^{G}}}
\newcommand{\localStorage}[1]{\ensuremath{S^{L}_{#1}}}

\newcommand{\balanceSub}[1]{\balanceSubSuper{#1}{}}
\newcommand{\balanceNormalizedSub}[1]{\ensuremath{\balanceNormalizedSubSuper{#1}{}}}

\newcommand{\protocolGenericName}{\ensuremath{\varPi}}

\newcommand{\mainChainName}[1]{\ensuremath{C_{#1}}}
\newcommand{\lastBlockOnMainChain}[1]{\ensuremath{\textit{last}\left(#1\right)}}
\newcommand{\lengthMainChain}[1]{\ensuremath{\textit{length}\left(#1\right)}}

\newcommand{\countableVariable}{\ensuremath{m}}

\newcommand{\balanceSubSuper}[2]{\ensuremath{B_{#1}^{#2}\xspace}}
\newcommand{\balanceNormalizedSubSuper}[2]{\ensuremath{b_{#1}^{#2}\xspace}}

\newcommand{\internalCurrency}{\ensuremath{\textit{ic}}}
\newcommand{\externalCurrency}{\ensuremath{\textit{ec}}}

\newcommand{\balanceOfOthers}{\ensuremath{B_{\userSet}^{\internalCurrency}\xspace}}
\newcommand{\balanceOfMinersInternal}{\ensuremath{B_{\minerSet}^{\internalCurrency}\xspace}}
\newcommand{\balanceOfMinersExternal}{\ensuremath{B_{\minerSet}^{\externalCurrency}\xspace}}
\newcommand{\balanceOfMinersAll}{\ensuremath{B_{\minerSet}\xspace}}

\newcommand{\strategyNameProtocol}[2]{\ensuremath{\sigma_{\text{#1}}^{\text{#2}}}}

\newcommand{\funcNameAllocateBudgt}{Allocate}
\newcommand{\funcNameGenerateBlock}{Generate}
\newcommand{\funcNamePublish}{Publish}

\newcommand{\funcNameBalance}{Bal}
\newcommand{\funcNameLongestChains}{LC}
\newcommand{\funcNameLongestCommonPrefix}{LCP}

\newcommand{\funcArgsAllocateBudgt}[1]{\globalStorage, \balanceSub{#1}}
\newcommand{\funcArgsGenerateBlock}[1]{\globalStorage, \localStorage{#1}}
\newcommand{\funcArgsPublish}[1]{\globalStorage, \localStorage{#1}}

\newcommand{\funcNameAllocateBudgtWithFont}{\textsf{\funcNameAllocateBudgt}_{i}()}
\newcommand{\funcNameGenerateBlockWithFont}{\textsf{\funcNameGenerateBlock}_{i}()}
\newcommand{\funcNamePublishWithFont}{\textsf{\funcNamePublish}_{i}()}

\newcommand{\funcNameBalanceWithFont}[1]{\textit{\funcNameBalance}\left(#1\right)}

\newcommand{\funcNameBalanceWithFontAndProtocol}[2]{\textit{\funcNameBalance}^{#2}(#1)}
\newcommand{\funcNameBalanceWithFontAndProtocolAndIndex}[3]{\textit{\funcNameBalance}^{#2}_{#3}(#1)}

\newcommand{\funcAllocateBudgt}[1]{\ensuremath{\textsf{\funcNameAllocateBudgt}_{#1}\left(\funcArgsAllocateBudgt{#1}\left(\epochIndex\right)\right)}}
\newcommand{\funcGenerateBlock}[2]{\ensuremath{\textsf{\funcNameGenerateBlock}_{#1}^{#2}\left(\funcArgsGenerateBlock{#1}\right)}}
\newcommand{\funcPublish}[1]{\ensuremath{\textsf{\funcNamePublish}_{#1}\left(\funcArgsPublish{#1}\right)}}

\newcommand{\funcBalanceWithProtocol}[2]{\ensuremath{\textit{\funcNameBalance}_{#1}^{#2}\left(\funcLongestCommonPrefix\right)}}

\newcommand{\minerSet}{\ensuremath{\textit{Miners}}\xspace}
\newcommand{\userSet}{\ensuremath{\textit{Users}}\xspace}

\newcommand{\funcLongestChains}[1]{\ensuremath{\textit{\funcNameLongestChains}_{#1}\left(\globalStorage\right)}}
\newcommand{\funcLongestChainsNotGlobal}{\ensuremath{\textit{\funcNameLongestChains}\left(S\right)}}

\newcommand{\funcLongestCommonPrefix}{\ensuremath{\textit{\funcNameLongestCommonPrefix}\left(\globalStorage\right)}}
\newcommand{\funcLongestCommonPrefixNotGlobal}{\ensuremath{\textit{\funcNameLongestCommonPrefix}\left(S\right)}}

\newcommand{\epochIndex}{\ensuremath{k}\xspace}

\newcommand{\weightLetter}{W}
\newcommand{\weight}[1]{\ensuremath{\weightLetter_{#1}}}
\newcommand{\weightNormalized}[1]{\ensuremath{\tfrac{\E\left[\weightWithChain{#1}{\mainChainName{}}\right]}{ \balanceNormalizedSub{#1} }}}
\newcommand{\weightNormalizedNormalized}[1]{\ensuremath{\tfrac{1}{ \epochLength \factor } \cdot  \tfrac{\E\left[\weightWithChain{#1}{\mainChainName{}}\right]}{ \balanceNormalizedSub{#1}  }}}
\newcommand{\weightWithChain}[2]{\ensuremath{\weightLetter_{#1}\left({#2}\right)}}

\newcommand{\countFactoredBlocksWithChain}[2]{\ensuremath{N_{#1}\left({#2}\right)}}

\newcommand{\countFactoredBlocksWithChainWithEpoch}[2]{\ensuremath{N_{#1}^{\epochIndex}\left({#2}\right)}}

\newcommand{\utilityfuncName}[1]{\ensuremath{\textit{U}}_{#1}\xspace}
\newcommand{\utilityfuncNameWithProtocol}[2]{\ensuremath{\textit{U}}_{#1}^{#2}\xspace}

\newcommand{\utilityNoIcfuncName}[2]{\ensuremath{\textit{U}}_{#1,\noIcStrategyName}^{#2}\xspace}
\newcommand{\utilityYesIcfuncName}[2]{\ensuremath{\textit{U}}_{#1,\defaultStrategyName}^{#2}\xspace}

\newcommand{\hebShortName}{\ensuremath{\textit{HEB}}\xspace}
\newcommand{\hebFullName}{\ensuremath{\textit{Hybrid Expenditure Blockchain}}\xspace}

\newcommand{\protocolNameOneNoPrefixShortName}{\ensuremath{\textit{PR}}\xspace}
\newcommand{\protocolNameOneNoPrefixFullName}{\ensuremath{\textit{PartialReward}}\xspace}

\newcommand{\protocolNameBitcoin}{\ensuremath{\textit{Nakamoto}}\xspace}
\newcommand{\protocolNameBitcoinHalf}{\ensuremath{\textit{NakamotoHalf}}\xspace}

\newcommand{\protocolNameOne}{\protocolNameOneNoPrefixShortName}

\newcommand{\protocolNameMandatory}{\ensuremath{\textit{MIE}}\xspace}
\newcommand{\protocolNameMandatoryFullName}{\ensuremath{\textit{MandatoryInternalExepnditure}}\xspace}

\newcommand{\ignorantMiningStrategyName}{\ensuremath{\textit{PoW-only}}\xspace}
\newcommand{\defaultStrategyName}{\ensuremath{\textit{prescribed}}\xspace}
\newcommand{\pettyCompliantStrategyName}{\ensuremath{\textit{petty}}\xspace}
\newcommand{\bestResponseStrategyName}{\ensuremath{\textit{best}}\xspace}
\newcommand{\noIcStrategyName}{\ensuremath{\textit{prescribed-no-ic}}\xspace}

\newcommand{\actionContent}{\ensuremath{
	\left\{
		\actionFieldOneName{},
		\actionFieldTwoName{}
	\right\}
}}

\newcommand{\actionFieldOneName}{\ensuremath{\textit{chain\_manipulation}}}
\newcommand{\actionFieldTwoName}{\ensuremath{\textit{block\_type}}}

\newcommand{\nodeContent}{\ensuremath{
	\left\{
		\fieldOneName{},
		\fieldTwoName{},
		\fieldThreeName{}
	\right\}
}}

\newcommand{\fieldOneName}{\ensuremath{\textit{secret\_chain}}}
\newcommand{\fieldTwoName}{\ensuremath{\textit{public\_chain}}}
\newcommand{\fieldThreeName}{\ensuremath{\textit{fork}}}

\newcommand{\protocolPropertyDecentralization}{\ensuremath{\textit{Size bias}}}
\newcommand{\protocolPropertyPermissiveness}{\ensuremath{\textit{Permissiveness}}}
\newcommand{\protocolPropertyAdversarial}{\ensuremath{\textit{Nash threshold}}}
\newcommand{\protocolPropertyPow}{\ensuremath{\textit{External expenses}}}
\newcommand{\protocolPropertySabotageAttack}{\ensuremath{\textit{Safety-violation threshold}}}
\newcommand{\protocolPropertyDoubleSpendAttack}{\ensuremath{\textit{Free safety-violation threshold}}}

\begin{document}

\title{Tuning PoW with Hybrid Expenditure}

\sloppy



\author{Itay Tsabary}
\email{sitay@technion.ac.il}
\affiliation{%
	\institution{Technion, IC3}
	\city{Haifa}
	\country{Israel}
}

\author{Alexander Spiegleman}
\email{sashaspiegelman@fb.com}
\affiliation{%
	\institution{Faceook Novi}
	\city{Menlo Park}
	\state{California}
	\country{USA}
}

\author{Ittay Eyal}
\email{ittay@technion.ac.il}
\affiliation{%
	\institution{Technion, IC3}
	\city{Haifa}
	\country{Israel}
}

\begin{abstract} 
	
	\emph{Proof of Work} (\emph{PoW}) is a Sybil-deterrence security mechanism. 
	It introduces an \emph{external cost} to a system by requiring computational effort to perform actions.  
	However, since its inception, a central challenge was to tune this cost. 
	Initial designs for deterring spam email and DoS attacks applied overhead equally to honest participants and attackers. 
	Requiring too little effort did not deter attacks, whereas too much encumbered honest participation. 
	This might be the reason it was never widely adopted. 
	
	Nakamoto overcame this trade-off in Bitcoin by distinguishing desired from malicious behavior and introducing \emph{internal rewards} for the former. 
	This solution gained popularity in securing cryptocurrencies and using the virtual internally-minted tokens for rewards. 
	However, in existing blockchain protocols the internal rewards fund (almost) the same value of external expenses. 
	Thus, as the token value soars, so does the PoW expenditure.
	Bitcoin PoW, for example, already expends as much electricity as Colombia or Switzerland. 
	This amount of resource-guzzling is unsustainable and hinders even wider adoption of these systems.

	In this work we present \emph{Hybrid Expenditure Blockchain} (\emph{HEB}), a novel PoW mechanism. 
	\emph{HEB} is a generalization of Nakamoto's protocol that enables tuning the external expenditure by introducing a complementary \emph{internal-expenditure} mechanism.
	Thus, for the first time, HEB decouples external expenditure from the reward value.

	We show a practical parameter choice by which \emph{HEB} requires significantly less external consumption compare to Nakamoto's protocol, its resilience against rational attackers is similar, and it retains the decentralized and permissionless nature of the system.
	Taking the Bitcoin ecosystem as an example, \emph{HEB} cuts the electricity consumption by half.

\end{abstract}

\maketitle

\ccsdesc[500]{Security and privacy~Distributed systems security}
\ccsdesc[500]{Applied computing~Digital cash}

\keywords{Blockchain; Proof of work; Cryptocurrency; Wastefulness;  Environmental impact }


\section{Introduction}
\label{sec:intro}

\emph{Permissionless} systems are susceptible to~\emph{Sybil attacks}~\cite{douceur2002sybil} where a single attacker can masquerade as multiple entities.
To mitigate such attacks, \emph{Proof of Work}~(\emph{PoW})~\cite{dwork1992pricing,jakobsson1999proofs,aspnes2005exposing} security schemes introduce \emph{external costs}, making attacks expensive.
To perform operations in a PoW system, users must provide a proof of computation, whose production requires resource expenditure.
This makes attacks like email spam~\cite{rosenberg2008session} and denial of service~\cite{wang2003defending,stebila2011stronger,back2002hashcash,aura2000resistant,juels1999client} prohibitively expensive, as they require many operations.
However, honest users are also subject to these costs, and the system cannot balance deterring adversarial behavior but not honest one~\cite{laurie2004proof}.

To circumvent this trade-off, Nakamoto~\cite{nakamoto2008bitcoin} suggested introducing \emph{internal rewards} for honest behavior (Table~\ref{table:comparison} summarizes this taxonomy).
Indeed, nowadays PoW is widely used to secure decentralized and permissionless cryptocurrencies like Bitcoin~\cite{nakamoto2008bitcoin} and Ethereum~\cite{buterin2013ethereum}.
These are replicated state-machines~\cite{borg1983message,lamport1984using,castro1999practical,lamport1978implementation} that facilitate monetary ecosystems of internally-minted \emph{tokens}, maintained by principals called \emph{miners}. 
Miners that follow the protocol are rewarded with tokens;
tokens are scarce, hence a market forms~\cite{lynn1991scarcity,brock1968implications,brock1992liberalization,cryptoslatePowMarketCap};
so, miners can sell their obtained tokens for fiat currency (e.g., USD, EUR), compensating them for their PoW expenses.

\begin{table}[!b]
	
	\newcolumntype{P}[1]{>{\centering\arraybackslash}p{#1}}
	\newcolumntype{M}[1]{>{\centering\arraybackslash}m{#1}}
	
	\begin{center}
		\begin{tabular}{| M{1.2cm} | M{2.2cm} | M{2.2cm} | M{1.5cm} |}
			\hline		
			\multirow[b]{2}{*}{\makecell{ \textbf{Internal}\\ \textbf{Rewards} }} & \multicolumn{3}{c|}{\textbf{Expenses}}  \\\cline{2-4}
			& \textit{Negligible} & \textit{External} & \textit{External \& Internal} \\
			\hline
			\multirow[]{2}{*}{\makecell{ \textit{Exist} }} & PoS Blockchains & PoW Blockchains & $\hebShortName$ \\
			& \cite{kiayias2017ouroboros,gilad2017algorand,ef2020eth2Pos,goodman2014tezos} & \cite{nakamoto2008bitcoin,buterin2013ethereum} & (this~work) \\ 
			\hline
			\multirow[]{2}{*}{\makecell{ \textit{Absent} }} & Open systems & Classical PoW &  \\
			& (e.g., email) & \cite{dwork1992pricing,juels1999client,aspnes2005exposing} & \\ 			
			\hline
		\end{tabular}
	\end{center}
	
	\caption{Security scheme rewards-expenses comparison.}
	\label{table:comparison}
\end{table}

To guarantee their security~\cite{garay2015backbone,pass2017analysis,decker2013propagation}, PoW cryptocurrencies moderate their operation rate by dynamically tuning the required computation difficulty to match miner capabilities~\cite{nakamoto2008bitcoin,kiffer2018better,noda2020economic,negy2020selfish}.
Consequently, the PoW expenditure directly depends on token values~\cite{mirkin2020bdos,becker2012integrity,ciaian2016economics,kroll2013economics}~--
higher token prices imply higher mining rewards, which draw more miners to participate, leading to more expended resources.
This results with the external PoW expenditure matching the internal mining rewards, and hence balances the overhead for honest participation with high attack costs.

Indeed, with exponentially-growing token values~\cite{bitcoin2020price,ether2020price}, the amount of resources spent on PoW mining has also been growing accordingly~\cite{blockchain2019difficulty,ether2019difficulty}.
Bitcoin PoW computations alone are responsible for about~$0.3\%$ of the global electricity consumption~\cite{de2020bitcoin,digiconomist2019BitcoinEnergyConsumptionIndex}, surpassing medium-sized countries like Colombia and Switzerland~\cite{cbeci}.
This level of resource guzzling is unsustainable~\cite{fairley2017blockchain,su9122214,ong2018plattsburghBanMining,wildau2018chinaShutterBitcoinMines,truby2018decarbonizing,de2019renewable,hern2019bitcoinenergyusageishugewecantaffordtoignoreit,eckel2020blackouts,greenberg2019energy,kamiya2019bitcoinEnergyUse,crs2019bitcoin},  bears a significant ecological impact~\cite{stoll2019carbon,houy2019rational,mora2018bitcoin,goodkind2020cryptodamages,bitcoinMiningCoal,lou2019thenextbigenvironmentalfight,feige2019bitcoindoesntincentivizegreenenergy,johnson2019thenegativeenvironmentalimpactofbitcoin,hunt2018canWePrevent}, and prevents adoption~\cite{bbc2021teslaNoLongerAcceptsBitcoin,bbc2020blackrockGoesGreen}.
Unfortunately, Nakamoto's mechanism directly incentivizes external expenditure at the same rate as of the internal rewards, and offers no means of reducing its external effects.

Previous work~(\S\ref{sec:related_work}) explored PoW alternatives for cryptocurrencies, notably focusing on~\emph{Proof of Stake} (\emph{PoS})~\cite{kiayias2017ouroboros,gilad2017algorand,ef2020eth2Pos}.
Such systems avoid the external resource expenditure by replacing the computational effort with internal token ownership requirements.
However, PoS systems operate, and are secured under, qualitatively-different assumptions.
Namely, new participants need to obtain tokens, which occurs only with the permission of current stakeholders.

We note that naive adjustments of the cryptocurrency minting rate do not reduce the external expenses; that a simple reduction of rewards hampers security; and that \emph{forcing} miners to internally-spend breaks the permissionless property of the system.
We review these options in the appendix.

In this work we present~$\hebFullName$ ($\hebShortName$), the first PoW protocol with lower external costs than its internal rewards.
Despite the reduced external expenditure, $\hebShortName$ provides similar security guarantees against rational attackers compared to the more-wasteful Nakamoto protocol.
$\hebShortName$ is tunable, allowing the system designer to optimize for desired properties.

		\subsection*{$\hebShortName$ Overview}
		\label{sec:intro_protocol}

The main challenge is to reduce the external expenses while keeping attack (and  also participation) costs high.
These objectives seem to contradict, as in previous work~\cite{dwork1992pricing,jakobsson1999proofs,aspnes2005exposing,nakamoto2008bitcoin,buterin2013ethereum} the participation costs are only external.

This is the main novelty of~$\hebShortName$~-- it is a generalization of Nakamoto's protocol that \emph{enables} and \emph{incentivizes} miners to forfeit system tokens as part of the mining process.
Miners that do so increase their rewards, resulting with this being the optimal behavior.
So, the external mining expenses in~$\hebShortName$ are lower than existing PoW blockchains, while the total expenses (internal and external) are the same.

Similarly to Bitcoin, $\hebShortName$ constructs a tree data structure of elements named \emph{blocks}, where the longest path (called \emph{main chain}) defines the system state.
Miners produce PoW to create blocks and broadcast them with a p2p network.
However, unlike Bitcoin, $\hebShortName$ considers epochs of blocks, and the mining rewards for each epoch are distributed only at its end.

In $\hebShortName$ there are two types of blocks miners can create~-- \emph{regular} or \emph{factored}, which have a~\emph{weight} attribute with values~$1$ and~$\factor > 1$, respectively.
The epoch rewards are distributed among the miners proportionally to the relative weight of their blocks (in the epoch).
Miners can always create regular blocks, but have to forfeit system tokens in epoch~$\epochIndex$ (e.g., with a designated transaction) to create factored blocks in epoch~$\epochIndex+1$.
This mechanism incentives miners to divert some of their total  participation budget internally to create factored blocks, reducing their external PoW expenditure.
The ratio between internal and external expenses is tuned with a parameter~$\posRatio$, determining the token expenditure required to create a factored block.

To maintain the total circulating token supply, the internally-expended tokens are distributed proportionally among all system entities (i.e., any token holder) at the epoch conclusion.
This redistribution maintains the token value as in a standard PoW cryptocurrency (e.g., Bitcoin), as well as the proportional purchasing power~\cite{core2017economy} (i.e., relative wealth) of all system entities.
The internally-spent tokens are redistributed using a novel redistribution technique, which might be of independent interest (e.g., for regulating transaction fees, cf.~\cite{buterin2020eip1559,roughgarden2020transaction}).


We emphasize that $\hebShortName$ draws ideas from PoS, prominently the utilization of system tokens for security, but the model assumptions, the solution, and the guarantees are distinct. 
In particular, $\hebShortName$ uses the standard PoW assumptions and miners expend (lose) their tokens for mining, whereas in PoS participants derive their power from maintaining ownership.

		\subsection*{$\hebShortName$ Analysis}
		\label{sec:intro_analysis}

Reasoning about $\hebShortName$ requires a substantial groundwork.
We begin by modeling the cryptocurrency ecosystem, starting with the relation of token price and supply, and following with the underlying data structures, participants, and execution~(\S\ref{sec:model}).
Our model is based on standard economic theories, where a commodity's price is inverse to its circulating supply; to the best of our knowledge, this is the first work to apply this classic modeling to cryptocurrency mining.

As in previous work~\cite{eyal2014majority,sapirshtein2016optimal,arvindcutoff,tsabary2018thegapgame,gervais2016security,tsabary2020mad,winzer2019temporary}, we consider a set of rational miners that optimize their revenues and an adversary who is willing to expend resources in order to attack.



We instantiate $\protocolNameBitcoin$'s protocol~(\S\ref{sec:bitcoin_protocol}) and use it for a comparison baseline, following with the formal presentation of~$\hebShortName$~(\S\ref{sec:protocol_v3}).

To compare and contrast~$\hebShortName$ with~$\protocolNameBitcoin$ we consider a variety of cryptocurrency metrics~(\S\ref{sec:protocol_v3_analysis}).
These include common security metrics, namely coalition resistance and tendency to encourage coalitions~\cite{eyal2014majority,sapirshtein2016optimal,eyal2015dilemma,tsabary2018thegapgame}.
We also introduce a new metric~-- \emph{external expenses}, measuring the resources spent on PoW. 
Instead of the binary metric permissioned/permissionless (classical-consensus-protocols/Nakamoto-blockchain, respectively), we introduce the continuous metric of~\emph{permissiveness}, describing the cost of joining the system. 

Finally, we tease apart the common safety-violation security metric into two.
We observe that safety-violating chain-reorganization attacks~\cite{redman2020bitcoingold51,bitcoin2020etc51,dcimit202051attacks,voell2020etc51} in existing PoW blockchains require high resource investment from the attacker; however, once successful, they completely refund themselves.
We therefore consider this type of attacks, as well as a sabotage variant where the attacker is not refunded.
We show that in~$\hebShortName$ the attack cost for the refunded variant is linear in the total expenses (secure as~$\protocolNameBitcoin$), and that the sabotage variant costs are linear in the external ones.

$\hebShortName$ includes several parameters for the system designer to tune.
As an example we present a specific choice of parameter values~(\S\ref{sec:concrete_instantiation_parameters}).
Choosing the prominent Bitcoin ecosystem as a reference point, we analyze $\hebShortName$ and show this parameter choice is practical, achieves strong security guarantees, and leads to the equivalent of reducing the entire electricity consumption of Denmark~\cite{de2020bitcoin,countryEnergyConsumption2020citypopulation}.

In summary, we expand the PoW design space by introducing internal expenses.
We present~$\hebShortName$~-- a PoW blockchain protocol with external expenses that are lower than its internal rewards.
We prove that~$\hebShortName$ offers similar security guarantees against rational attackers compared to pure PoW solutions, and show it can significantly reduce the latter's ecological impact.

\section{Related Work}
\label{sec:related_work}

In Nakamoto's blockchain and all subsequent PoW protocols we are aware of, the incentives equal the value of the generated cryptocurrency tokens (and fees). 
We are not aware of previous work tuning PoW expenditure in cryptocurrencies~-- the main focus of this work.

We proceed to survey PoS and analysis approaches.
We defer to Appendix~\ref{app:related_pow_alternatives} a discussion of permissioned~\cite{castro1999practical,kotla2007zyzzyva,kapitza2012cheapbft,gueta2018sbft,yin2019hotstuff,baudet2018state,cachin2016architecture} and trusted hardware~\cite{zhang2017rem,chen2017security} solutions that make qualitatively-stronger assumptions; protocols that expend different external resources rather than electricity~\cite{stewart2012proof,p4titan2014slimcoin,DBLP:journals/iacr/KarantiasKZ19,miller2014permacoin,bentov2017tortoise}, for which~$\hebShortName$ applies equally well; and protocols~\cite{king2012peercoin,ali2020pox} with several types of internal tokens that do not achieve incentive compatibility nor reduced external expenses.

\paragraph{Proof of Stake}

$\hebShortName$ and PoS are fundamentally different: the latter limits miner participation to those with stake in the system, i.e., miners who own tokens; the former does not.
Moreover, in PoS the Sybil-deterrence~\cite{douceur2002sybil} is due to the cost of acquiring and holding the system tokens, which the participants maintain throughout the system execution.
In contrast, $\hebShortName$ relies on PoW, and the participants \emph{spend} the internal currency.

PoS systems like Algorand~\cite{kiayias2017ouroboros}, Ouroboros~\cite{gilad2017algorand}, Tezos~\cite{goodman2014tezos} and Ethereum2~\cite{ef2020eth2Pos} are designed and analyzed under different assumptions than PoW.
Their security is measured with respect to the number of owned tokens rather than expended resources. 
They assume a new participant wanting to join the system can acquire (or alternatively, lock as a collateral~\cite{kiayias2017ouroboros,goodman2014tezos,ef2020eth2Pos}) as many system tokens as she can afford.
That is, existing system miners authorize transactions that introduce new system miners, even if these result in a state less favorable from their perspective.
Additionally, to combat \emph{long-range attacks}~\cite{deirmentzoglou2019survey,gavzi2018stake} and \emph{nothing-at-stake}~\cite{poelstra2014distributed,bonneau2015sok}, these systems assume users voluntarily delete deprecated data~\cite{kiayias2017ouroboros}, or assume users remain online for extended periods~\cite{gilad2017algorand,goodman2014tezos,ef2020eth2Pos}.

In contrast,~$\hebShortName$ is PoW-based, and newly-joining miners do not require the cooperation of existing miners to join.
It is also resistant to said long-range attacks and the nothing-at-stake problems, and hence does not rely on voluntarily data deletion or user persistence.

A parallel work~\cite{ganesh2020virtual} draws ideas from~$\hebShortName$, suggesting to emulate PoW over PoS.
The main idea is that the stake used for the consensus degrades over time and usage, mimicking the external expenditures of PoW systems.
However, as built atop of PoS, it also requires the aforementioned assumptions.

\paragraph{Proof of Work Analysis}

We use the standard  techniques~\cite{eyal2014majority,sapirshtein2016optimal,nayak2016stubborn,arvindcutoff,tsabary2018thegapgame,liao2017incentivizing,zhang2019lay,gervais2016security,hou2019squirrl} to analyze~$\hebShortName$'s security and incentive compatibility.
The evaluation metrics used are a formalization of previous ones presented by said work, and also include definition of new ones regarding the external expenditure and permissiveness level.
To the best of our knowledge we are the first to define, evaluate and optimize for such metrics.

Chen et al.~\cite{chen2019axiomaticApproach} define and analyze desired properties of reward allocation schemes in PoW cryptocurrencies. 
Their work focuses on the reward of a single block, and does not consider environmental impact nor malicious miners.
We note that $\hebShortName$'s reward allocation rule is incentive-compatible and Sybil-resistant, satisfying these desired properties.

	\section{Cryptocurrency Model}
	\label{sec:model}

We present a model for an abstract blockchain system, instantiated with a cryptocurrency protocol.
This allows us to consider different instantiations, namely $\protocolNameBitcoin$ and $\hebShortName$.
We first define the monetary value of system tokens using an exogenous reference-point fiat currency~(\S\ref{sec:model_exchange_rate}).
We follow by presenting the blockchain, the participating entities and how the system derives its state~(\S\ref{sec:model_state}).
We then define how a cryptocurrency protocol instantiates that system, defining an internal system currency based on the blockchain~(\S\ref{sec:model_cryptocurrency}), and explain how the system makes progress~(\S\ref{sec:model_execution}).

		\subsection{Cryptocurrency Economics}
		\label{sec:model_exchange_rate}

The external expenditure of a PoW cryptocurrency system depends on the rewards it grants miners and mining costs.
We note that mining rewards are internal while PoW costs are external, hence we first define the relation, or the exchange rate, of the two.

The reward is an amount of the system's internal currency~$\internalCurrency$ (e.g., Bitcoin, Ether), and the external cost is an amount of an external currency~$\externalCurrency$ (e.g., EUR, USD, RMB). 
We assume the external currency has a market capital orders of magnitude larger than that of the internal currency~\cite{fiatmarketcap}, and it effectively represents real values.

According to the prominent macroeconomic theories~\cite{friedman1989quantity,allais1966restatement,lucas1980two,keynes2018general,clarke1988keynesianism,hall1989political,lynn1991scarcity,brock1968implications,brock1992liberalization} there is a linear relation between \emph{supply} and \emph{price levels} of commodities in equilibrium states.
Applying to cryptocurrency tokens, the steady-state real value of a single~$\internalCurrency$ token, measured in~$\externalCurrency$, is inverse to the circulating supply of~$\internalCurrency$.
This comfortably settles with intuition~--~$\internalCurrency$ tokens have value due to their scarcity, and if they were twice less scarce then their real value would have been half.

This relation was previously used to analyze a single-shot minting of tokens in Initial Coin Offerings (ICOs)~\cite{burniske2017cryptoassets,buterin2017price,catalini2016some,catalini2018initial}, but to the best of our knowledge, this work is the first to apply it to the ongoing token minting process.

Throughout this work, when comparing two cryptocurrencies, we assume their market caps are equal.
This implies that if the two cryptocurrencies have the same minting rate then their respective token real values are identical.
That is why naively changing a cryptocurrency minting rate does not affect the external expenditure~(see Appendix~\ref{app:naive_protocols}).

We assume there is an instantaneous and commission-free exchange service of~$\externalCurrency$ and~$\internalCurrency$, where the exchange rate matches token real value.
Unless explicitly mentioned, we assume this exchange is available to all participating entities.

To simplify presentation we normalize the price level so the exchange rate is one.
We often sum~$\externalCurrency$ and~$\internalCurrency$, meaning the sum of their values in real terms.

		\subsection{Blockchain and System Principals}
		\label{sec:model_state}

The system comprises a shared~\emph{global storage}~$\globalStorage$, a~\emph{scheduler}, and two types of~\emph{entities}: system~\emph{users}, and principals maintaining the system named~\emph{miners}.

\paragraph{Global storage}
The global storage is an append-only set containing elements called~\emph{blocks}.
Each block includes a reference to another block and data generated by system entities, with the only exception being a so-called~\emph{genesis block} that contains neither.
The global storage~$\globalStorage$ initially contains only the genesis block, thus defining a directed tree data structure rooted at the genesis block.

We refer to paths in the data structure starting at the genesis block as~\emph{chains}.
We denote for any chain~$\mainChainName{}$ its last block by~$\lastBlockOnMainChain{\mainChainName{}}$ and its length  by~$\lengthMainChain{\mainChainName{}}$.
We divide a chain~$\mainChainName{}$ to~\emph{epochs} of~$\epochLength$ blocks.
For any~$\epochIndex \in \N$, epoch~$\epochIndex$ includes the block series~$\left[\epochLength \epochIndex + 1,\epochLength\left(\epochIndex+1\right)\right]$ in~$\mainChainName{}$.

A function~$\funcLongestChainsNotGlobal{}$ returns the set of the~\emph{longest chains} in a block set~$S$.
A function~$\funcLongestCommonPrefixNotGlobal$ returns the~\emph{longest common prefix} chain of~$\funcLongestChainsNotGlobal{}$.
We refer to~$\funcLongestCommonPrefix$ as the~\emph{main chain}.

\paragraph{Miners and the scheduler}

As common~\cite{eyal2014majority,sapirshtein2016optimal,nayak2016stubborn,arvindcutoff,tsabary2020mad} we assume that the sets of entities are static during an epoch execution, that is, entities do not join or leave during the course of an epoch.
Denote by~$\minerSet\left(\epochIndex\right)$ and~$\userSet\left(\epochIndex\right)$ the miners and users of epoch~$\epochIndex$, respectively.

Each miner~$i$ has a local storage~$\localStorage{i}$ accessible only to her.
Like~$\globalStorage$,~$\localStorage{i}$ is an append-only block set.
The scheduler invokes miners, allowing them to create blocks in their local storage, and to publish their local blocks by copying them to~$\globalStorage$.
We denote by~$\countFactoredBlocksWithChainWithEpoch{i}{\mainChainName{}}$ the number of blocks in epoch~$\epochIndex$ created by miner~$i$ on chain~$\mainChainName{}$.
For presentation simplicity, we assume the main chain at an epoch beginning remains a prefix of the main chain throughout the entire epoch.
Note this does not rule out the main chain changing during an epoch, but only that its initial prefix does not.

\paragraph{State}
Entities derive the \emph{system state} by parsing the global storage~$\globalStorage$ according to the block order of the main chain. 

System entities might choose to infer the state based on a chain prefix, excluding potentially-volatile suffixes~\cite{garay2015backbone,pass2017analysis}, such as in the case of multiple longest chains.
Such considerations are outside the scope of this work.

\subsection{Instantiating a Cryptocurrency Protocol}
\label{sec:model_cryptocurrency}

The system is instantiated with a cryptocurrency protocol~$\protocolGenericName$ that defines a currency internal to the system,~$\internalCurrency$.
The protocol~$\protocolGenericName$ maps all its internal tokens to system entities through a function~$\funcNameBalanceWithFontAndProtocol{\mainChainName{}}{\protocolGenericName}$, taking as input a chain~$\mainChainName{}$.
The~$\funcNameBalanceWithFontAndProtocol{\mainChainName{}}{\protocolGenericName}$ function returns a vector where each element~$\funcNameBalanceWithFontAndProtocolAndIndex{\mainChainName{}}{\protocolGenericName}{i}$ is the number of tokens mapped to entity~$i$.
When the context is clear, we often omit the protocol name~$\protocolGenericName$ and simply write~$\funcNameBalanceWithFontAndProtocol{\mainChainName{}}{}$.

We say the real value of tokens mapped to an entity is her~$\internalCurrency$~\emph{balance}, and note the total number of tokens is the sum of all balances.
The protocol mints~$\rewardPerBlockWithEpoch \cdot \epochLength$ new tokens at the end of each epoch~$\epochIndex$, and we often omit the epoch index when it is clear from context.
This means the number of tokens is fixed throughout any epoch~$\epochIndex$, and increases when epoch~$\epochIndex + 1$ begins.

Formally, let~$\tokenCount_{j}$ be the total number of tokens on chain~$\mainChainName{}$ when~$\lengthMainChain{\mainChainName{}} = j$.
So, for any~$j \in \left[\epochLength \epochIndex, \epochLength \left(\epochIndex+1\right)-1\right]$ it holds that~$\tokenCount_{j } = \tokenCount_{\epochLength \epochIndex}$, and that~$\tokenCount_{\epochLength \left(\epochIndex + 1\right) } - \tokenCount_{\epochLength \epochIndex} = \rewardPerBlockWithEpoch \cdot \epochLength$.
The protocol~$\protocolGenericName$ maps the newly-minted tokens to entities using~$\funcNameBalanceWithFont{\mainChainName{}}$.

\begin{myNote}
	To avoid temporal inconsistencies, protocols often make newly-minted tokens available only after sufficiently many other blocks are created.
	For example, Bitcoin makes tokens minted in block~$j$ available only after block~$j+100$.
	We neglect such mechanisms for simplicity, but they demonstrate that delayed availability of minted tokens is both acceptable and practical.
\end{myNote}

\paragraph{Miner economics}

Aside from their~$\internalCurrency$ balances, miners also own~$\externalCurrency$.
We use the terms~\emph{internal} and~\emph{external balances} to distinguish the different currency holdings, and simply~\emph{balance} to describe their aggregate value.

Miners expend all their balance on system maintenance. 
In practice, a principal can split its balance, using some of it as a miner, and the rest as a user.
We model such principals as two separate entities~-- a miner that spends all its balance on maintenance, and a user that holds the rest.

For any epoch~$\epochIndex$ we denote by~$\balanceSubSuper{i}{\externalCurrency}\left(\epochIndex\right)$ and~$\balanceSubSuper{i}{\internalCurrency}\left(\epochIndex\right)$ the initial external and internal balances of each miner~$i \in \minerSet\left(\epochIndex\right)$, respectively.
We denote by~$\balanceSub{i}\left(\epochIndex\right) \triangleq \balanceSubSuper{i}{\externalCurrency}\left(\epochIndex\right) + \balanceSubSuper{i}{\internalCurrency}\left(\epochIndex\right)$ the initial balance of miner~$i$.

We denote the accumulated internal and external miner balances by~$\balanceOfMinersInternal\left(\epochIndex\right) \triangleq \sum_{j \in \minerSet}{\balanceSubSuper{j}{\internalCurrency}\left(\epochIndex\right)}$ and~$\balanceOfMinersExternal\left(\epochIndex\right) \triangleq \sum_{j \in \minerSet}{\balanceSubSuper{j}{\externalCurrency}\left(\epochIndex\right)}$, respectively.
We also denote the total balances of all miners by~$\balanceOfMinersAll\left(\epochIndex\right)$ and of all users by~$\balanceOfOthers\left(\epochIndex\right)$.
We denote the external and internal relative balances of miner~$i$ by~$\balanceNormalizedSubSuper{i}{\externalCurrency}\left(\epochIndex\right)\triangleq\tfrac{\balanceSubSuper{i}{\externalCurrency}\left(\epochIndex\right)}{\balanceOfMinersExternal\left(\epochIndex\right)}$ and~$\balanceNormalizedSubSuper{i}{\internalCurrency}\left(\epochIndex\right)\triangleq\tfrac{\balanceSubSuper{i}{\internalCurrency}\left(\epochIndex\right)}{\balanceOfMinersInternal\left(\epochIndex\right)}$, respectively, and her relative balance by~$\balanceNormalizedSub{i}\left(\epochIndex\right)\triangleq\tfrac{\balanceSub{i}\left(\epochIndex\right)}{\balanceOfMinersAll\left(\epochIndex\right)}$. 

We assume the value of expended resources by the miners on system maintenance in a single epoch~$\epochIndex$ is much smaller than the system market cap.
That is, the balance of all miners is negligible compared to that of all users, i.e.,~$\balanceOfMinersAll\left(\epochIndex\right) \ll \balanceOfOthers\left(\epochIndex\right)$.
This holds both in Bitcoin and in Ethereum where $\tfrac{\balanceOfMinersAll\left(\epochIndex\right)}{\balanceOfOthers\left(\epochIndex\right)}$ is approximately~$3.4\cdot10^{-7}$~\cite{coinmarketcap2020circulationBTC} and~$8.95\cdot10^{-9}$~\cite{coinmarketcap2020circulationETH}, respectively.

\subsection{Execution}     
\label{sec:model_execution}

Initially, the global storage~$\globalStorage$ contains only the genesis block, and each miner~$i$ has an empty local storage~$\localStorage{i} = \emptyset$.
The state variables (like the global and local storage) change over time, but we omit indexing as it is clear from the context.

The system progresses is orchestrated by the scheduler~(Alg.~\ref{alg:scheduler}).

Epoch~$\epochIndex$ begins when~$\lengthMainChain{\funcLongestCommonPrefix{}} =  \epochLength \epochIndex$.
First,~\emph{nature} sets all miner balances~(lines~\ref{alg:scheduler_nature_start}--\ref{alg:scheduler_nature_end}), modeling miners joining and leaving at epoch transitions.

The scheduler then lets miners set their internal and external balances using the exchange service, achieving their preferred balance of the two.
We use the term~\emph{allocate} to describe this action, and say miner~$i$ allocates her balance~$\balanceSub{i}\left(\epochIndex\right)$ with the invocation of the $\funcAllocateBudgt{i}$ function, returning a tuple of her internal and external balances~$\left\langle\balanceSubSuper{i}{\internalCurrency}\left(\epochIndex\right),\balanceSubSuper{i}{\externalCurrency}\left(\epochIndex\right)\right\rangle$. 
So, for each miner~$i$ the scheduler invokes $\funcAllocateBudgt{i}$~(lines~\ref{alg:scheduler_allocate_start}--\ref{alg:scheduler_allocate_end}).

Note that modeling changes in the miner set and balance allocations at epoch transitions is for presentation simplicity; these occur throughout the system execution.


The rest of the epoch execution progresses in steps, until the main chain~$\funcLongestCommonPrefix{}$ is extended by~$\epochLength$ blocks (lines~\ref{alg:scheduler_extention_step_start}--\ref{alg:scheduler_extention_step_end}).
Each step begins with the scheduler selecting a single miner at random, proportionally to her relative~\emph{external expenditure}, that is~$\forall i \in \minerSet : \Pr \left( \text{scheduler selects } i \right) = \balanceNormalizedSubSuper{i}{\externalCurrency}\left(\epochIndex\right)$~(line~\ref{alg:scheduler_picks_miner}).
Similarly to previous work~\cite{eyal2014majority,arnosti2018bitcoin,sapirshtein2016optimal,nayak2016stubborn}, these steps represent a standard PoW mechanism and its logical state changes, and entities have synchronous access to the global storage.

The scheduler invokes the selected miner~$i$'s function $\funcGenerateBlock{i}{\protocolGenericName}$, returning a newly generated block, and adds it to miner~$i$'s local storage~$\localStorage{i}$~(line~\ref{alg:scheduler_miner_creates_block}).
The protocol~$\protocolGenericName$ states block validity rules in~$\funcNameBalanceWithFontAndProtocol{\mainChainName{}}{\protocolGenericName}$, and invalid blocks do not affect the system state.
Creating an invalid block or not creating one at all is sub-optimal and we only consider miners who avoid doing so.

Next, the scheduler lets each miner~$i$ publish her blocks by invoking $\funcPublish{i}$, returning a subset of her previously-private local blocks.
The scheduler adds the returned blocks to the global storage, and repeats this process until all miners do not wish to publish any more blocks (lines~\ref{alg:scheduler_publication_start}--\ref{alg:scheduler_publication_end}).
The publication loop is used to capture strategic-block-release behaviors~\cite{eyal2014majority,sapirshtein2016optimal,nayak2016stubborn}.

The cryptocurrency protocol~$\protocolGenericName$ includes implementations of~$\funcAllocateBudgt{i}$, $\funcGenerateBlock{i}{\protocolGenericName}$, and $\funcPublish{i}$ that each miner~$i$ should follow.
We refer to the tuple of three implementations as the~\emph{prescribed strategy} and denote it by~$\strategyNameProtocol{\defaultStrategyName}{\protocolGenericName}$.
The protocol~$\protocolGenericName$ is therefore a tuple of the balance function~$\textit{\funcNameBalance}^{\protocolGenericName}$ and a prescribed strategy~$\strategyNameProtocol{\defaultStrategyName}{\protocolGenericName}$.
Note that~$\protocolGenericName$ cannot force miners to follow~$\strategyNameProtocol{\defaultStrategyName}{\protocolGenericName}$.

\begin{algorithm}[t]
	\SetNoFillComment
	\SetAlgoNoLine 
	\SetAlgoNoEnd 
	\DontPrintSemicolon 
	\caption{Scheduler in epoch~$\epochIndex$} 
	\label{alg:scheduler} 
	\tcc{Initial storage state}
	\Input{$\globalStorage$ such that~$\lengthMainChain{\funcLongestCommonPrefix{}} =  \epochLength \epochIndex$}\label{alg:scheduler_initial_storage}
	\BlankLine

	
	\tcc{Setting miner balances}
	\For{$i \in \minerSet\left(\epochIndex\right)$}{
		\label{alg:scheduler_nature_start}
		$\balanceSubSuper{i}{}\left(\epochIndex\right) \gets v \in \mathbb{R}_{\ge 0} \text{, chosen by \textit{nature}}$ such that~$\balanceOfMinersAll\left(\epochIndex\right) \ll \balanceOfOthers\left(\epochIndex\right)$\; 
	}\label{alg:scheduler_nature_end}
	\BlankLine

	\tcc{Miner balance allocation}
	\For{$i \in \minerSet\left(\epochIndex\right)$}{
		\label{alg:scheduler_allocate_start}
		$\left\langle\balanceSubSuper{i}{\internalCurrency}\left(\epochIndex\right),\balanceSubSuper{i}{\externalCurrency}\left(\epochIndex\right)\right\rangle \gets \funcAllocateBudgt{i}$\; 
	}\label{alg:scheduler_allocate_end}
	\BlankLine
	
	\tcc{Main chain extention}	
	\While{~$\lengthMainChain{\funcLongestCommonPrefix{}} < \epochLength\left(\epochIndex+1\right)$}{ 		\label{alg:scheduler_extention_step_start}
		\tcp{Block generation}
		$i \gets$ miner index chosen at random,~$\forall j \in \minerSet\left(\epochIndex\right) : \Pr\left(i=j\right) = \balanceNormalizedSubSuper{j}{\externalCurrency}\left(\epochIndex\right)$\;\label{alg:scheduler_picks_miner}
		$\localStorage{i} \gets \localStorage{i} \cup \left\{\funcGenerateBlock{i}{\protocolGenericName}\right\} 		$\;\label{alg:scheduler_miner_creates_block}
		\BlankLine
		\tcp{Block publication}

		$\text{blocks\_to\_publish} \gets \emptyset$\;\label{alg:scheduler_publication_start}
		\Do{$\text{blocks\_to\_publish} \ne \emptyset$}{
			\label{alg:scheduler_publication_start_loop}
			$\globalStorage \gets \globalStorage \cup \text{blocks\_to\_publish}$\; 
			$\text{blocks\_to\_publish} \gets \emptyset$\;
			\For{$i \in \minerSet\left(\epochIndex\right)$}{
				$\text{blocks\_to\_publish} \gets \text{blocks\_to\_publish} \cup \funcPublish{i}$\;
			}
		}
		\label{alg:scheduler_publication_end}
	}\label{alg:scheduler_extention_step_end}

\end{algorithm}

\begin{figure}[!t]
	
	\makeatletter
	\newcommand{\removelatexerror}{\let\@latex@error\@gobble}
	\makeatother
	
	\removelatexerror

	\begin{minipage}[t]{.23\textwidth}
		\vspace*{-\baselineskip}
		\begin{algorithm}[H]
			\SetNoFillComment
			\SetAlgoNoLine 
			\SetAlgoNoEnd 
			\DontPrintSemicolon 
			\caption{$\strategyNameProtocol{i,\defaultStrategyName}{\protocolNameBitcoin}$} 
			\label{alg:bitcoin_protocol_strategy}

			\Fn{\FAllocateBudgt{\funcArgsAllocateBudgt{i}}}{
				\Return~$\left\langle 0,\balanceSub{i} \right\rangle$\;
			}
			\Fn{\FGenerateBlock{\funcArgsGenerateBlock{i}}}{
				$\mainChainName{} \gets \text{ uniformly from } \funcLongestChains{}$
				
				$\text{pointer} \gets \lastBlockOnMainChain{\mainChainName{}}$\;
				\Return NewBlock(pointer)\;
			}
			
			\Fn{\FPublish{\funcArgsPublish{i}}}{
				\Return All unpublished blocks\;
			}
		\BlankLine
		\BlankLine
		\BlankLine		
		\BlankLine			
		\BlankLine		
		\BlankLine	
		\BlankLine			
		\BlankLine		
		\BlankLine		
		\BlankLine	
		\BlankLine			
		\BlankLine		
		\BlankLine		
		\BlankLine	
		\BlankLine		

		\end{algorithm}
	\end{minipage}%
	\hfill
	\begin{minipage}[t]{.23\textwidth}
		\vspace*{-\baselineskip}
		\begin{algorithm}[H]
			\SetNoFillComment
			\SetAlgoNoLine 
			\SetAlgoNoEnd 
			\DontPrintSemicolon 
			\caption{$\strategyNameProtocol{i,\defaultStrategyName}{\hebShortName}$} 
			\label{alg:v3_protocol_strategy} 
			
			\SetKwIF{If}{ElseIf}{Else}{if}{}{else if}{else}{end if}%
			
			\Fn{\FAllocateBudgt{\funcArgsAllocateBudgt{i}}}{
				\Return~$\left\langle \posRatio \balanceSub{i}, \oneMinusPosRatioBrackets \balanceSub{i} \right\rangle~$\;
			}
			\Fn{\FGenerateBlock{\funcArgsGenerateBlock{i}}}{
				$\mainChainName{} \gets \text{ uniformly from } \funcLongestChains{}$\;
				$\text{pointer} \gets \lastBlockOnMainChain{\mainChainName{}}$\;
				\If{
					$\posRatio=0$ \KwOr $\countFactoredBlocksWithChain{i}{\mainChainName{}} < \floor*{\tfrac{\balanceSubSuper{i}{\internalCurrency}}{\posRatio \cdot \rewardPerBlock}}$}
				{~$\text{type} \gets \textit{factored}$}
				\Else{~$\text{type} \gets \textit{regular}$}
				\Return NewBlock(pointer, type)\;
			}
			\Fn{\FPublish{\funcArgsPublish{i}}}{
				\Return All unpublished blocks\;
			}
		\end{algorithm}
	\end{minipage}

	\vspace{-1.5\baselineskip}

\end{figure}

\section{Nakamoto Protocol}
\label{sec:bitcoin_protocol}

As an example and to serve as a baseline, we instantiate an epoch-based~$\protocolNameBitcoin$ protocol (used with~$\epochLength=1$ in Bitcoin~\cite{nakamoto2008bitcoin}, Bitcoin Cash~\cite{bch2018site}, Litecoin~\cite{litecoin2013site}, etc.) in our model.

The balance function of~$\protocolNameBitcoin$ awards each miner~$i$ with~$\rewardPerBlock$ tokens per block she created in the epoch, and a total of~$\epochLength \cdot \rewardPerBlock$ new tokens are minted.
Hence, the balance of each miner~$i$ at epoch conclusion is~$\funcBalanceWithProtocol{i}{\protocolNameBitcoin} = \countFactoredBlocksWithChain{i}{\funcLongestCommonPrefix{}} \rewardPerBlock$.

The prescribed strategy~$\strategyNameProtocol{\defaultStrategyName}{\protocolNameBitcoin}$~(Alg.~\ref{alg:bitcoin_protocol_strategy}) states that each miner~$i$ allocates her balance~$\balanceSubSuper{i}{\externalCurrency} = \balanceSub{i}$ and~$\balanceSubSuper{i}{\internalCurrency} = 0$, extends the longest chain, and publishes her blocks immediately.
In case of multiple longest chains,~$\strategyNameProtocol{\defaultStrategyName}{\protocolNameBitcoin}$ picks uniformly-at-random\footnote{Bitcoin defines a different tie-breaking rule~--- pick the first longest chain the miner became aware of.
Therefore its security guarantees vary, depending on the underlying network assumptions.
As in previous work~\cite{eyal2016bitcoin,kogias2016byzcoin,sapirshtein2016optimal}, we avoid such assumptions by considering the uniformly-at-random variation.}.

\section{$\hebShortName$ Protocol}
\label{sec:protocol_v3} 

We are now ready to present~$\hebShortName$.
Briefly, it incentivizes miners to expend their balances internally by enabling miners who do so to create higher-rewarding blocks.
Two parameters,~$\posRatio \in \left[0,1\right)$ and~$\factor \in \mathbb{R}_{>1}$, control the reward distribution mechanism.
We detail the different block types, the reward distribution mechanism, and the prescribed strategy.

\paragraph{Block types} 

Each block has a type, determined at its creation~-- either \emph{regular} or \emph{factored}.
During the epoch miner~$i$ can create regular blocks whenever the scheduler invokes~$\funcNameGenerateBlockWithFont$.
In contrast, creating a factored block on~$\mainChainName{}$ requires an expenditure of~$\posRatio \cdot \rewardPerBlock$ in~$\internalCurrency$ by miner~$i$ at the previous epoch;
recall we model such expenditure as if it occurs at the start of current epoch.

Consequently, if~$\posRatio>0$ then miner~$i$ can create at most~$\floor*{\tfrac{\balanceSubSuper{i}{\internalCurrency}}{\posRatio\cdot \rewardPerBlock}}$ factored blocks in an epoch on chain~$\mainChainName{}$.
$\hebShortName$ assigns a~\emph{weight} to each block according to its type, and factored and regular blocks have weights of~$\factor$ and~$1$, respectively.

\paragraph{Reward distribution} 

$\hebShortName$ distributes the~$\epochLength \rewardPerBlock$ minted tokens among the miners in proportional to their contributed block weights.
Denote by~$\weightWithChain{i}{\mainChainName{}}$ the total weight of blocks created in the epoch by miner~$i$ on chain~$\mainChainName{}$.
So, miner~$i$ gets~$\tfrac{\weightWithChain{i}{\mainChainName{}}}{\sum_{j \in \minerSet}{\weightWithChain{j}{\mainChainName{}}}} \epochLength \rewardPerBlock$ tokens for her created blocks.

$\hebShortName$ distributes the internal expenses~$\balanceOfMinersInternal$ among all system entities (i.e., including users) proportionally to their~$\internalCurrency$ balances at the epoch beginning.
So, miner~$i$ receives~$\tfrac{\balanceSubSuper{i}{\internalCurrency}}{\balanceOfMinersInternal + \balanceOfOthers} \balanceOfMinersInternal$ tokens from the redistribution.

In summary, the balance of miner~$i$ with the epoch conclusion is~$\funcBalanceWithProtocol{i}{\hebShortName} = \tfrac{\weightWithChain{i}{\funcLongestCommonPrefix{}}}{\sum_{j \in \minerSet}{\weightWithChain{j}{\funcLongestCommonPrefix{}}}} \epochLength \rewardPerBlock + \tfrac{\balanceSubSuper{i}{\internalCurrency}}{\balanceOfMinersInternal + \balanceOfOthers} \balanceOfMinersInternal$.

\begin{myNote}
	Token redistribution is linear in the relative token balance in order to maintain the relative purchasing power~\cite{core2017economy} of each entity.
	This is in line with previous work~\cite{nakamoto2008bitcoin,chen2019axiomaticApproach}, and the natural approach in a permissionless and decentralized setting: super-linear allocation incentivizes centralization; sub-linear results with the rich splitting as several entities; both are undesired~\cite{chen2019axiomaticApproach}.
\end{myNote}

\paragraph{Prescribed strategy}
The prescribed strategy~$\strategyNameProtocol{\defaultStrategyName}{\hebShortName}$~(Alg.~\ref{alg:v3_protocol_strategy}) states that miners allocate their balance with ratio~$\posRatio$, and create factored blocks up to their internal balance limitation.
Formally, miner~$i$ allocates~$\balanceSubSuper{i}{\internalCurrency} = \posRatio\balanceSub{i}$ and~$\balanceSubSuper{i}{\externalCurrency} = \oneMinusPosRatioBrackets\balanceSub{i}$.
If~$\posRatio=0$ then the miner creates all blocks as factored, and if~$\posRatio>0$ then only the first~$\floor*{\tfrac{\balanceSubSuper{i}{\internalCurrency}}{\posRatio\cdot \rewardPerBlock }}$ ones.
As in $\protocolNameBitcoin$, miner~$i$ points her created blocks to a uniformly-at-random selected chain from~$\funcLongestChains{}$, and publishes them immediately.

\begin{myNote}
	Setting~$\posRatio = 0$ enables miners to create all blocks as factored, and setting~$\factor = 1$ removes motivation to create any factored blocks at all.
	In both cases there is only one practical block type, reducing $\hebShortName$ to $\protocolNameBitcoin$.
\end{myNote}

Appendix~\ref{app:practical_considerations} brings practical implementation aspects of~$\hebShortName$ -- shortening epochs, utilizing a pure PoW ramp-up period to create a sufficiently-large currency circulation, and addressing discretization issues.

	\section{Evaluation}
	\label{sec:protocol_v3_analysis} 

We now evaluate~$\hebShortName$, showing how parameter choices affect its properties.
For that, we formalize the cryptocurrency system as a game played among system entities, striving to maximize their rewards~(\S\ref{sec:model_block_creation_as_a_game}).
We use~$\protocolNameBitcoin$ as a baseline, highlighting parameter choices that result with significantly lower PoW expenditure while limiting undesirable side-effects.

To realize such a comparison we first need to define its criteria.
Hence, throughout the rest of this section we present cryptocurrency evaluation metrics, each followed by its evaluations of~$\protocolNameBitcoin$ and~$\hebShortName$.

We consider previous security metrics~\cite{eyal2014majority,sapirshtein2016optimal,nayak2016stubborn,arvindcutoff,liao2017incentivizing,gervais2016security,mirkin2020bdos,chen2019axiomaticApproach,rosenfeld2011analysis,pass2017analysis,garay2015backbone} regarding the~\emph{incentive compatibility} of a system~(\S\ref{sec:protocolPropertyDecentralization} and~\S\ref{sec:protocolPropertyAdversarial}); refine the common safety-violation security metric~\cite{rosenfeld2014hashrateAnalysis,karame2012fast,bonneau2016buy,mccorry2018smart,judmayer2019pay}, measuring attack~\emph{costs}~(\S\ref{sec:protocolPropertyDoubleSpendAttack}); generalize the binary permissioned/permissionless notion~\cite{nakamoto2008bitcoin,baudet2018state} to a continuous metric~(\S\ref{sec:protocolPropertyPermissiveness}); and conclude with a new natural metric for external expenses~(\S\ref{sec:protocolPropertyPow}).

\subsection{Block Creation as a Game}
\label{sec:model_block_creation_as_a_game}

The model gives rise to a game, played for the duration of a single epoch~$\epochIndex$.
The players are the miners, each with her epoch balance as an input.

We define the utility of miner~$i$ for an epoch~$\epochIndex$ as her expected cryptocurrency holdings at the conclusion of the epoch, i.e.,~$\utilityfuncNameWithProtocol{i}{}\left(\epochIndex\right)\triangleq\E\left[\funcBalanceWithProtocol{i}{}\right]$ when~$\lengthMainChain{\funcLongestCommonPrefix} = \epochLength\left(k+1\right)$.

As commonly done in the analysis of cryptocurrency protocols~\cite{eyal2014majority,arvindcutoff,huberman2017monopoly,tsabary2020mad,sapirshtein2016optimal}, we assume that during an epoch the system is quasi-static, where all miners participate and the total profit is constant.
In operational systems miners participate for a positive profit~\cite{mirkin2020bdos,kristoufek2020bitcoin,yaish2020pricing}, but discussing the required return-on-investment ratio for such behavior is out the scope of this work, and we arbitrarily assume it to be~$0$~\cite{fiat2019energy,goren2019mindTheMining,tsabary2018thegapgame}.
Accordingly, the sum of all miner utilities equals the overall miner balances, that is,
\begin{equation}
\label{eq:model_total_utilities_equals_expenses}
\balanceOfMinersAll\left(\epochIndex\right) = \sum_{i \in \minerSet\left(\epochIndex\right)}{\utilityfuncName{i}\left(\epochIndex\right)}\,\,.
\end{equation}

We normalize the number of newly-minted tokens in epoch~$\epochIndex$ per block to be one, meaning~$\rewardPerBlockWithEpoch = 1$, and a total of~$\epochLength$ tokens are created in the epoch.

The mining strategy space comprises choosing the balance allocation ratio, what blocks to generate, and when to publish them, i.e., implementations of $\funcAllocateBudgt{}$, $\funcGenerateBlock{}{\protocolGenericName}$ and $\funcPublish{}$.

\paragraph*{Example: Nakamoto}
We demonstrate the compatibility of our definitions and modeling with previous results~\cite{narayanan2016bitcoin} regarding~$\protocolNameBitcoin$.

We consider a scenario where all miners follow~$\strategyNameProtocol{\defaultStrategyName}{\protocolNameBitcoin}$.
So, all miner balances are in~$\externalCurrency$ and consequently~$\balanceNormalizedSubSuper{i}{}=\balanceNormalizedSubSuper{i}{\externalCurrency}$.
Additionally, all miners extend the longest chain, so there is only a single one~$\left(\left|\funcLongestChains{}\right| = 1\right)$, which we denote by~$\mainChainName{}$, and it follows that~$\mainChainName{} = \funcLongestCommonPrefix$.

We note the scheduler picks at each step a miner proportionally to her relative external balance~(Alg.~\ref{alg:scheduler}, line~\ref{alg:scheduler_picks_miner}).
We can consider each pick as a Bernoulli trial where miner~$i$ is picked with success probability of~$\balanceNormalizedSub{i}$.
So, the number of blocks a miner~$i$ creates in an epoch is binomially distributed~$\countFactoredBlocksWithChain{i}{\mainChainName{}} \sim \operatorname{Bin}\left({\epochLength, \balanceNormalizedSub{i}}\right)$.

Therefore,~$\E\left[\countFactoredBlocksWithChain{i}{\mainChainName{}}\right] = \balanceNormalizedSubSuper{i}{\externalCurrency} \epochLength$, and the utility of miner~$i$ is $\utilityfuncNameWithProtocol{i}{\protocolNameBitcoin} = \balanceNormalizedSubSuper{i}{\externalCurrency} \epochLength$, matching previous analysis~\cite{nakamoto2008bitcoin}.
Summing for all miners and applying the balance-income equation (Eq.~\ref{eq:model_total_utilities_equals_expenses}) yields~$\balanceOfMinersAll = \balanceOfMinersExternal = \epochLength$, and the expected cost to create each block is~$1$, matching its reward.

		\subsection{$\protocolPropertyDecentralization$}
		\label{sec:protocolPropertyDecentralization}

Cryptocurrency security relies on having multiple, independent miners, none of which has control over the system~\cite{chen2019axiomaticApproach,rosenfeld2011analysis,pass2017fruitchains,pass2017analysis,garay2015backbone}.
For that, these systems strive to distribute their rewards in a way that is~\emph{size-indifferent}~\cite{chen2019axiomaticApproach}, meaning that miners get relative reward matching their relative balances, and hence have no incentive to coalesce.
The metric~$\protocolPropertyDecentralization$ measures how well a protocol satisfies this desideratum when all miners follow the prescribed strategy. 
Unlike the other metrics, it is evaluated for a specific balance distribution (i.e., a specific allocation).

Formally, assume a balance distribution and that each miner~$i$ with relative balance~$\balanceNormalizedSubSuper{i}{}$ follows~$\strategyNameProtocol{\defaultStrategyName}{\protocolGenericName}$.
The utility of such miner is~$\utilityfuncNameWithProtocol{i}{\protocolGenericName}$, and her relative utility is~$\tfrac{\utilityfuncNameWithProtocol{i}{\protocolGenericName}}{\sum_{j \in \minerSet}{\utilityfuncNameWithProtocol{j}{\protocolGenericName}}}$.
We define~$\protocolPropertyDecentralization$ to be the maximal difference of each miner's relative balance and relative utility, that is,~$\protocolPropertyDecentralization\triangleq\max\limits_{i \in \minerSet}{\left|\balanceNormalizedSubSuper{i}{} - \tfrac{\utilityfuncNameWithProtocol{i}{\protocolGenericName}}{\sum_{j \in \minerSet}{\utilityfuncNameWithProtocol{j}{\protocolGenericName}}}\right|}$.

Systems strive for~$\protocolPropertyDecentralization$ to be minimal, as higher values indicate more disproportionate shares.
Preferably,~$\protocolPropertyDecentralization = 0$, indicating all miners get reward proportionally to their balance.

In practice, there is an inherent advantage for having a larger relative balance.
For example, a larger miner is competing against a smaller portion of the network, reducing the chance of its blocks being forked of the main chain~\cite{neudecker2019short,gencer2018decentralization,shahsavari2019theoretical,lewenberg2015cooperative,miller2015nonoutsourceable}.
Additionally, economy-of-scale optimizations allow miners with higher balances to operate more efficiently~\cite{arnosti2018bitcoin}.
These considerations are outside the scope of this work, and current literature does not provide a specific number for comparison.
So, although $\protocolPropertyDecentralization = 0$ is a theoretical desideratum, systems like Bitcoin successfully operate even with non-zero values.

		\subsubsection*{\textbf{Nakamoto}}

Recall that in~$\protocolNameBitcoin$ the utility of each miner~$i$ is~$\utilityfuncNameWithProtocol{i}{\protocolNameBitcoin} =\balanceNormalizedSubSuper{i}{\externalCurrency} \epochLength$, meaning~$\tfrac{\utilityfuncNameWithProtocol{i}{\protocolNameBitcoin{}}}{\sum_{j \in \minerSet}{\utilityfuncNameWithProtocol{j}{\protocolNameBitcoin{}}}} = \balanceNormalizedSubSuper{i}{}$ and in our model~$\protocolPropertyDecentralization = 0$.
This matches previous analysis~\cite{nakamoto2008bitcoin}.

		\subsubsection*{\textbf{HEB}}

We now analyze~$\hebShortName$'s~$\protocolPropertyDecentralization$.
We begin with a preliminary analysis of the prescribed strategy~(\S\ref{sec:protocol_v3_prescribed_strategy_analysis}), deriving the expected miner utilities.

We follow with formalizing and proving two lemmas~(\S\ref{app:v3_decentralization_epoch_factor}): the first showing a required and sufficient condition to achieve~$\protocolPropertyDecentralization = 0$; the second showing this condition is met for sufficiently-large epoch lengths (i.e., larger~$\epochLength$).

We conclude with concrete number instantiations~(\S\ref{app:v3_decentralization_balance}), showing that $\protocolPropertyDecentralization$ improves (decreases) with longer epochs (larger~$\epochLength$) and a smaller factor (smaller~$\factor$) value, while being independent of~$\posRatio$.
We also show more balanced distributions have lower~$\protocolPropertyDecentralization$, but note these are not under the control of the system designer.
Considering practical parameter choices, we show that even for an extreme balance distribution, $\hebShortName$ achieves $\protocolPropertyDecentralization < 0.3\%$.
In a similar, yet balanced scenario,~$\protocolPropertyDecentralization= 0$.

			\subsubsection{\textbf{Prescribed Strategy Analysis}}
			\label{sec:protocol_v3_prescribed_strategy_analysis}

We show that if all miners follow the prescribed strategy then
\begin{equation}
\label{eq:v3_balance_miners}
\balanceOfMinersAll = \epochLength\,\,
\end{equation}
and utility of each miner $i$ is
\begin{equation}
\label{eq:v3_expected_income}
\utilityfuncNameWithProtocol{i}{\hebShortName} = \tfrac{\E\left[\weightWithChain{i}{\mainChainName{}}\right]}{\sum_{j \in \minerSet}{\E\left[\weightWithChain{j}{\mainChainName{}}\right]}} \epochLength\,\,.
\end{equation}

This analysis requires the following steps.
First we show the redistributed internal currency a miner receives is negligible, allowing us to focus on the minting reward.
For that we analyze the number of blocks a miner creates.
Then, we derive her conditional total block weight, that is, her total block weight conditioned on the number of blocks she creates.
We proceed to derive her expected total block weight, and conclude with finding her utility.
Note that as all miners follow~$\strategyNameProtocol{\defaultStrategyName}{\hebShortName}$ there is a single longest chain~$\mainChainName{}$.

\paragraph{Negligible internal distribution rewards}
According to~$\strategyNameProtocol{\defaultStrategyName}{\hebShortName}$ each miner~$i$ allocates her balance such that~$\balanceSubSuper{i}{\internalCurrency} = \posRatio \balanceSub{i}$.
Summing the reward all miners receive and substituting total rewards with total balance (Eq.~\ref{eq:model_total_utilities_equals_expenses}) yields~$\balanceOfMinersAll = \epochLength + \tfrac{\left(\posRatio{} \balanceOfMinersAll\right)^2}{\posRatio \balanceOfMinersAll + \balanceOfOthers}$, and as such~$\balanceOfMinersAll \left( 1 - \tfrac{\posRatio{}^2 \balanceOfMinersAll}{\posRatio \balanceOfMinersAll + \balanceOfOthers}\right) = \epochLength$.
Recall that~$\posRatio < 1$ and~$\balanceOfMinersAll \ll \balanceOfOthers$, therefore~$\posRatio{}^2 \balanceOfMinersAll \ll \balanceOfOthers$ (cf. \S\ref{sec:model_exchange_rate}) and~$\tfrac{\posRatio{}^2 \balanceOfMinersAll}{\posRatio \balanceOfMinersAll + \balanceOfOthers}$ is negligible.
We get that~$\balanceOfMinersAll = \epochLength$~(Eq.~\ref{eq:v3_balance_miners}).

\paragraph{Number of blocks}
Each miner~$i$ allocates her balance such that~$\balanceSubSuper{i}{\externalCurrency} = \oneMinusPosRatioBrackets \balanceSub{i}$ and therefore the~$\externalCurrency$ ratio of miner~$i$ equals her balance ratio:~$\balanceNormalizedSubSuper{i}{\externalCurrency} = \tfrac{\balanceSubSuper{i}{\externalCurrency}}{\balanceOfMinersExternal} = \tfrac{\oneMinusPosRatioBrackets \balanceSub{i}}{ \oneMinusPosRatioBrackets \balanceOfMinersAll } = \balanceNormalizedSub{i}$.

Recall that the scheduler selects miner~$i$ to generate a block by her relative external expenses~$\balanceNormalizedSubSuper{i}{\externalCurrency}$.
Epochs are of length~$\epochLength$ and the number of blocks miner~$i$ creates in an epoch,~$\countFactoredBlocksWithChain{i}{\mainChainName{}}$, is binomially distributed~$\countFactoredBlocksWithChain{i}{\mainChainName{}} \sim \operatorname{Bin}\left({\epochLength, \balanceNormalizedSub{i}}\right)$.
Therefore, the probability that miner~$i$ creates exactly~$\countableVariable$ blocks in an epoch is~$\Pr\left(\countFactoredBlocksWithChain{i}{\mainChainName{}} = \countableVariable \right) = \binom{\epochLength}{\countableVariable} \cdot {\left(\balanceNormalizedSub{i}\right)}^\countableVariable \cdot {\left(1-\balanceNormalizedSub{i}\right)}^{\epochLength-\countableVariable}$, and her expected number of blocks is~$\E\left[\countFactoredBlocksWithChain{i}{\mainChainName{}}\right]=\epochLength \balanceNormalizedSub{i}$.

By~$\strategyNameProtocol{\defaultStrategyName}{\hebShortName}$ each miner~$i$ allocates her balance such that~$\balanceSubSuper{i}{\internalCurrency} = \posRatio \balanceSubSuper{i}{}$.
The required internal balance per factored block is~$\posRatio$, so miner~$i$ can create at most~$\floor*{\balanceSubSuper{i}{}}$ factored blocks.

Recall~$\balanceSubSuper{i}{} = \balanceNormalizedSubSuper{i}{} \cdot \balanceOfMinersAll$ and~$\balanceOfMinersAll = \epochLength$, and we get the maximal number of factored blocks miner~$i$ can create in an epoch is~$\floor*{\epochLength \balanceNormalizedSubSuper{i}{}}$.

For simplicity we assume that all miners have balances such that~$\epochLength \balanceNormalizedSub{i} \in \mathbb{N}$.
We note this implies the minimal relative balance of a miner to participate in the system.
However, because~$\epochLength$ is large, this limitation only applies to miners with very low, arguably insignificant, relative balances. E.g., if~$\epochLength = 1000$ then this limitation prevents miners with lower than 0.001 relative balance to participate, whom are irrelevant for any practical concern~\cite{arnosti2018bitcoin}.

We now find the miner~$i$'s total block weight, conditioned on the number of blocks she created.

\paragraph{Conditional block weight}
Assume miner~$i$ created~$\countableVariable$ blocks, and recall she can create at most~$\epochLength \balanceNormalizedSub{i}$ factored ones.

If~$0 \le \countableVariable \le \epochLength$ then miner~$i$ creates all her blocks as factored, each contributing block weight of~$\factor$, resulting with an accumulated block weight of~$\countableVariable \factor$.

However, if~$\epochLength \balanceNormalizedSubSuper{i}{} < \countableVariable \le$ then miner~$i$ creates only~$\epochLength \balanceNormalizedSub{i}$ factored blocks, while the rest,~$\countableVariable - \epochLength \balanceNormalizedSubSuper{i}{}$ are regular ones.
Consequently, her accumulated block weight is~$\epochLength \balanceNormalizedSub{i} \factor + \countableVariable - \epochLength \balanceNormalizedSub{i}$.

In summary, miner~$i$'s block weight assuming she created~$\countFactoredBlocksWithChain{i}{\mainChainName{}} = \countableVariable{}$ blocks is
\begin{equation}
\label{eq:v3_weight_conditional_variable}
\weight{i} \left(\mainChainName{} | \countableVariable{}\right) = 
\begin{dcases}
\countableVariable \factor, 								  & 0 \le  \countableVariable \le \epochLength \balanceNormalizedSubSuper{i}{} \\
\epochLength \balanceNormalizedSubSuper{i}{} \factor + \countableVariable - \epochLength \balanceNormalizedSubSuper{i}{}   & \epochLength \balanceNormalizedSubSuper{i}{} < \countableVariable \le \epochLength
\end{dcases}\,\,.
\end{equation}

\paragraph{Expected block weight}
It follows the expected block weight of miner~$i$ is~$\E\left[\weightWithChain{i}{\mainChainName{}}\right] =  \sum_{\countableVariable=0}^{\epochLength}{	\Pr\left(\countFactoredBlocksWithChain{i}{\mainChainName{}} = \countableVariable \right)  \weight{i} \left(\mainChainName{} | \countFactoredBlocksWithChain{i}{\mainChainName{}} = \countableVariable{}\right)}$\footnote{We note that the Chernoff~\cite{chernoff1952measure} and factorial moment~\cite{philips1995moment} bounds apply for~$\countFactoredBlocksWithChain{i}{\mainChainName{}}$, provide an exponentially decreasing (with respect to~$\epochLength$) bound on its distribution tail.
For a specific~$\epochLength$ value we can use the cumulative distribution and survival functions to find the error probability.
For example, for a miner~$i$ with relative balance~$\balanceNormalizedSubSuper{i}{}=0.3$ and~$\epochLength = 1000$, the probabilities for a negative and positive relative errors of 10\% are 0.02 and 0.018, respectively.}.

\paragraph{Expected income}
The utility based on the~$\funcNameBalance$ implementation is~$\utilityfuncNameWithProtocol{i}{\hebShortName} =  \tfrac{\E\left[\weightWithChain{i}{\mainChainName{}}\right]}{\sum_{j \in \minerSet}{\E\left[\weightWithChain{j}{\mainChainName{}}\right]}}   \epochLength  +  \dfrac{\posRatio^2 \balanceSubSuper{i}{}\balanceOfMinersAll}{\posRatio \balanceOfMinersAll + \balanceOfOthers}$.
Note that for similar considerations~$\tfrac{\posRatio{}^2 \balanceSub{i} \balanceOfMinersAll}{\posRatio \balanceOfMinersAll + \balanceOfOthers}$ is negligible, resulting with 
\begin{equation*}
\utilityfuncNameWithProtocol{i}{\hebShortName} =  \dfrac{\E\left[\weightWithChain{i}{\mainChainName{}}\right]}{\sum_{j \in \minerSet}{\E\left[\weightWithChain{j}{\mainChainName{}}\right]}}   \epochLength    \,\,,
\end{equation*}
which is exactly Eq.~\ref{eq:v3_expected_income}.

			\subsubsection{\textbf{Theoretical Analysis}}
			\label{app:v3_decentralization_epoch_factor}

We prove necessary and sufficient conditions to achieve~$\protocolPropertyDecentralization = 0$, and show they hold for sufficiently long epochs, i.e., large~$\epochLength$ values.

First, we note that in~$\hebShortName$ the relative utility of any miner~$i$ is~$\tfrac{\utilityfuncNameWithProtocol{i}{\hebShortName{}}}{\sum_{j \in \minerSet}{\utilityfuncNameWithProtocol{j}{\hebShortName{}}}} = \tfrac{\E\left[\weightWithChain{i}{\mainChainName{}}\right]}{\sum_{j \in \minerSet}{\E\left[\weightWithChain{j}{\mainChainName{}}\right]}}$ (summing Eq.~\ref{eq:v3_expected_income} for all miners), and hence we can consider~$\protocolPropertyDecentralization = \max\limits_{i \in \minerSet}{\left|\balanceNormalizedSubSuper{i}{} - \tfrac{\E\left[\weightWithChain{i}{\mainChainName{}}\right]}{\sum_{j \in \minerSet}{\E\left[\weightWithChain{j}{\mainChainName{}}\right]}}\right|}$.

We now present and prove two lemmas.
The first shows that all miners have the same expected block weight normalized by their relative balance iff~$\protocolPropertyDecentralization = 0$; the second that for a large~$\epochLength$ value all miners have the same normalized expected block weight.

\begin{lemma}
	\label{lemma:zero_decentralization_equal_nes}
	In~$\hebShortName$, if all miners follow~$\strategyNameProtocol{\defaultStrategyName}{\hebShortName}$, then~$\forall i,j \in \minerSet : \weightNormalized{i} = \weightNormalized{j}$ iff~$\protocolPropertyDecentralization = 0$.
\end{lemma}

\begin{proof}
	First, assume~$\forall i,j \in \minerSet : \weightNormalized{i} = \weightNormalized{j}$.
	It follows~$\forall i \in \minerSet :  \E\left[\weightWithChain{i}{\mainChainName{}}\right] = \tfrac{\balanceNormalizedSub{i}}{\balanceNormalizedSub{1}} \E\left[\weightWithChain{1}{\mainChainName{}}\right]$.
	Therefore~$\sum_{j \in \minerSet}{\E\left[\weightWithChain{j}{\mainChainName{}}\right]} = \tfrac{1}{\balanceNormalizedSub{1}} \E\left[\weightWithChain{1}{\mainChainName{}}\right]$, and consequently~$\tfrac{\E\left[\weightWithChain{i}{\mainChainName{}}\right]}{\sum_{j \in \minerSet}{\E\left[\weightWithChain{j}{\mainChainName{}}\right]}} = \balanceNormalizedSub{i}$.
	It immediately follows that~$\protocolPropertyDecentralization = 0$.
	
	Now assume~$\protocolPropertyDecentralization = 0$, meaning~$\tfrac{\E\left[\weightWithChain{i}{\mainChainName{}}\right]}{\sum_{j \in \minerSet}{\E\left[\weightWithChain{j}{\mainChainName{}}\right]}} = \balanceNormalizedSub{i}$ or~$\tfrac{\E\left[\weightWithChain{i}{\mainChainName{}}\right]}{\balanceNormalizedSub{i}} = \sum_{j \in \minerSet}{\E\left[\weightWithChain{j}{\mainChainName{}}\right]}$.
	The last equation holds for any miner~$i$, thus pose a set of linear equations, with the solution being that~$\forall i,j \in \minerSet : \weightNormalized{i} = \weightNormalized{j}$.
\end{proof}

\begin{lemma}
	\label{lemma:longer_epoch_converges_nes_to_one}
	In~$\hebShortName$, if all miners follow~$\strategyNameProtocol{\defaultStrategyName}{\hebShortName}$, then~$\forall i \in \minerSet : \lim\limits_{\epochLength \to \infty}  \weightNormalized{i} = \epochLength \factor$ .
\end{lemma}

\begin{proof}
	By the law of large numbers we get that~$\countFactoredBlocksWithChain{i}{\mainChainName{}}$ converges to its expected value,~$\lim\limits_{\epochLength \to \infty}  \countFactoredBlocksWithChain{i}{\mainChainName{}} = \epochLength \balanceNormalizedSub{i}$.
	Consequently, the expected block weight of each miner~$i$ (Eq.~\ref{eq:v3_weight_conditional_variable}) is~$\lim\limits_{\epochLength \to \infty}  \E\left[\weightWithChain{i}{\mainChainName{}}\right] = \epochLength \balanceNormalizedSubSuper{i}{} \factor$, and therefore~$\lim\limits_{\epochLength \to \infty}  \weightNormalized{i} = \epochLength \factor$.
\end{proof}

It directly follows from the two previous lemmas that with a sufficiently large~$\epochLength$ value~$\hebShortName$ achieves~$\protocolPropertyDecentralization = 0$, as stated in the following corollary:
\begin{corollary}
	In~$\hebShortName$, if all miners follow~$\strategyNameProtocol{\defaultStrategyName}{\hebShortName}$, then~$\lim\limits_{\epochLength \to \infty} \protocolPropertyDecentralization = 0$.
\end{corollary}

			\subsubsection{\textbf{Parameter Instantiation}}
			\label{app:v3_decentralization_balance}

We now evaluate the~$\protocolPropertyDecentralization$ for different~$\epochLength$,~$\factor$ and~$\balanceNormalizedSub{}$ values.
As all miners follow~$\strategyNameProtocol{\defaultStrategyName}{\hebShortName}$, then~$\protocolPropertyDecentralization$ is unaffected by~$\posRatio$.

We proceed as follows.
We numerically calculate~$\weightNormalized{i}$ for various~$\factor$,~$\epochLength$ and~$\balanceNormalizedSub{i}$ values.
We present our results in Fig.~\ref{fig:analysis_decentralization}, normalized by~$\tfrac{1}{\epochLength \factor}$ for comparison purposes.
Although we present results for specific configurations, we assert that different parameter values yield the same qualitative results.

\begin{figure}[!t]
	
	\centering

	\begin{subfigure}[b]{\VERTICALSUBFIGSCALE\textwidth}
		\resizebox{\textwidth}{!}{{\renewcommand\normalsize{\LARGE}\normalsize			\input{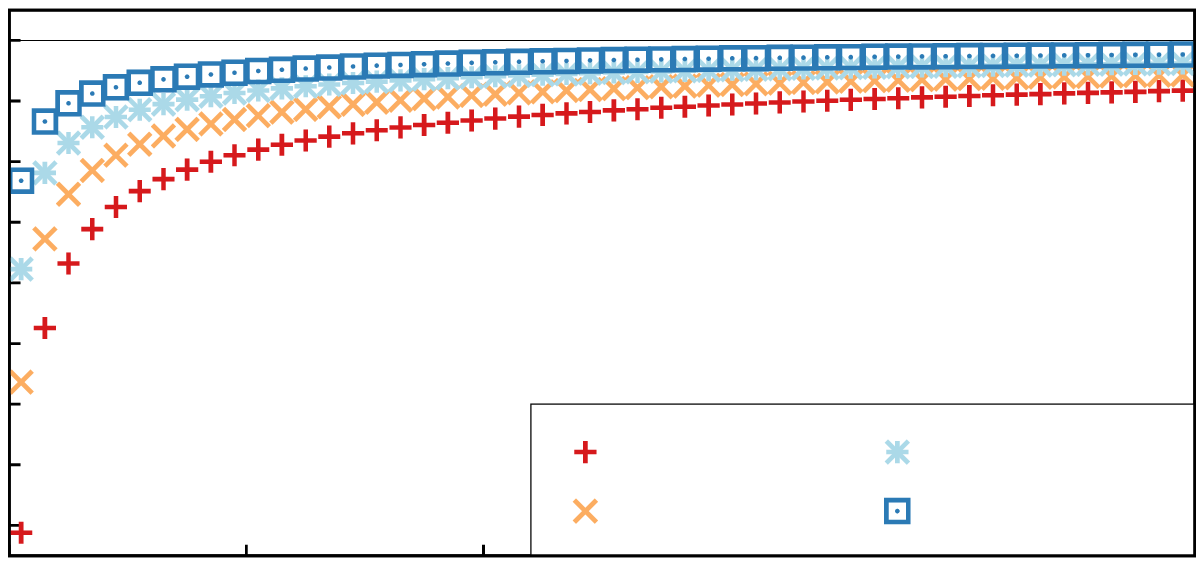}}}
		\vspace{-2\baselineskip}
		\caption{~$\factor{} = 20$. }
		\label{subfig:decentralization_analysis_epoch}
	\end{subfigure}%
	
	\begin{subfigure}[b]{\VERTICALSUBFIGSCALE\textwidth}
		\resizebox{\textwidth}{!}{{\renewcommand\normalsize{\LARGE}\normalsize			\input{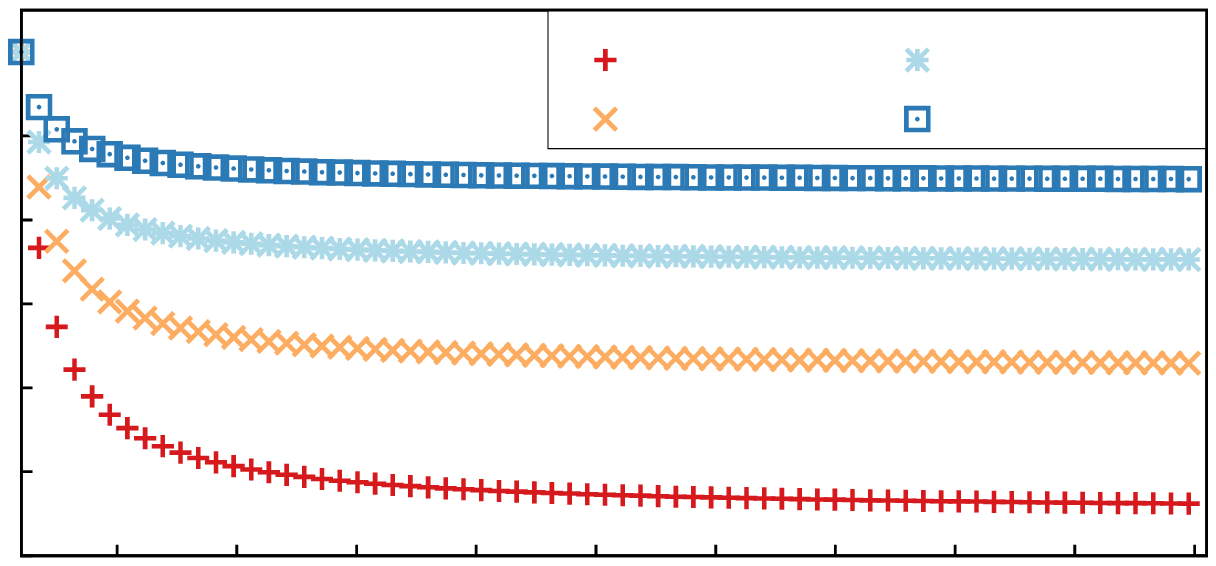}}}
		\vspace{-2\baselineskip}
		\caption{~$\epochLength = 1000$. }
		\label{subfig:decentralization_analysis_factor}
	\end{subfigure}%
	
	\vspace{-0.8\baselineskip}
	\caption{$\weightNormalizedNormalized{i}$ for~$\posRatio$ and~$\epochLength$ values. }
	\label{fig:analysis_decentralization}
	\vspace{-1\baselineskip}	
	
\end{figure}

Fig.~\ref{subfig:decentralization_analysis_epoch} shows for a fixed~$\factor = 20$ the value of~$\weightNormalizedNormalized{i}$ as a function of~$\epochLength$.
As expected (Lemma~\ref{lemma:longer_epoch_converges_nes_to_one}),~$\weightNormalized{i}$ approaches 1 as~$\epochLength$ grows, leading towards~$\protocolPropertyDecentralization = 0$.
However, for any fixed~$\epochLength$ value, miners of different~$\balanceNormalizedSub{i}$ have different~$\weightNormalizedNormalized{i}$, resulting with~$\protocolPropertyDecentralization > 0$, matching  Lemma~\ref{lemma:zero_decentralization_equal_nes}.

We also illustrate the effect of~$\factor$ on~$\weightNormalizedNormalized{i}$.
Fig.~\ref{subfig:decentralization_analysis_factor} shows~$\weightNormalizedNormalized{i}$ for~$\epochLength = 1000$ as function of~$\factor$.
At the region of lower~$\factor$ values, increasing~$\factor$ also increases the difference of~$\weightNormalizedNormalized{i}$ for different~$\balanceNormalizedSub{i}$.
However, as~$\factor$ becomes larger, then~$\weightNormalizedNormalized{i}$ tends towards a constant and the difference for different~$\balanceNormalizedSub{i}$ remains fixed.
This is expected, as for larger~$\factor$ values the term~$
\epochLength \balanceNormalizedSubSuper{i}{} \factor + \countableVariable - \epochLength \balanceNormalizedSubSuper{i}{}$ in Eq.~\ref{eq:v3_weight_conditional_variable} is dominated by~$
\epochLength \balanceNormalizedSubSuper{i}{} \factor$, and the expected weight becomes linear in~$\factor$.

We dedicate the rest of this section to analyze how different balance distributions affect miners' utilities and~$\protocolPropertyDecentralization$.
We consider various settings of at most~$5$ miners with epoch length of~$\epochLength = 1000$ blocks and~$\factor = 20$.

For each setting we numerically calculate~$\protocolPropertyDecentralization$ and present it, along with its respective balance distribution, in Table~\ref{tab:analysis_epsilon_values}.
We choose these specific settings to demonstrate~$\protocolPropertyDecentralization$ both balanced and extreme distributions.

\begin{table}[!t]
	\small	
	\centering
	\begin{tabular}{ |c|c|c|c|c||c| } 
		\hline
		\multicolumn{5}{|c||}{Balance distribution} &
		\multirow{2}{*}{$\protocolPropertyDecentralization$} \\
		\cline{1-5}
		$\balanceNormalizedSub{1}$ &~$\balanceNormalizedSub{2}$ &~$\balanceNormalizedSub{3}$ &~$\balanceNormalizedSub{4}$ &~$\balanceNormalizedSub{5}$ &   \\
		\hline		
		0.20 & 0.80 & - & - & - & \textbf{0.0029} \\
		\hline
		0.10 & 0.15 & 0.20 & 0.20 & 0.35 & \textbf{0.0025} \\
		\hline
		0.20 & 0.40 & 0.40 & - & - & \textbf{0.0015} \\
		\hline
		0.20 & 0.20 & 0.30 & 0.30 & - & \textbf{0.0007} \\
		\hline
		0.20 & 0.20 & 0.20 & 0.20 & 0.20 & \textbf{0.0000} \\
		\hline
	\end{tabular}
	
	\caption{$\protocolPropertyDecentralization$ when~$\epochLength = 1000$ and~$\factor = 20$. Hyphens represent miners not present in the configuration.} 
	\label{tab:analysis_epsilon_values} 
	\vspace{-3.1\baselineskip}
\end{table}

Table~\ref{tab:analysis_epsilon_values} shows that more extreme balance distributions results in higher~$\protocolPropertyDecentralization$.
For instance, consider the setting with only two miners where~$\balanceNormalizedSub{1} = 0.2$ and~$\balanceNormalizedSub{2} = 0.8$. 
This setting leads to the highest value of~$\protocolPropertyDecentralization = 0.0029$.
Note that this is an unrealistic setting, presented only as an example for a highly-uneven distribution.
Even in this extreme scenario miner~$1$ has a degradation of less than~$0.3\%$ in her relative utility.
More balanced settings lead to lower~$\protocolPropertyDecentralization$ values.

In summary, even a highly-unbalanced distribution results in minor deviations from proportional rewards.
Increasing~$\epochLength$ and decreasing~$\factor$ enables the system designer to reduce these deviations.


\subsection{$\protocolPropertyAdversarial$}
\label{sec:protocolPropertyAdversarial}		

Recall protocol~$\protocolGenericName$ provides a prescribed strategy~$\strategyNameProtocol{\defaultStrategyName}{\protocolGenericName}$ that miners individually choose whether to follow.
The protocol properties rely on miners following this strategy~\cite{eyal2014majority,sapirshtein2016optimal,nayak2016stubborn,arvindcutoff,tsabary2018thegapgame,zhang2019lay,gervais2016security,tsabary2020mad,mirkin2020bdos}, hence it should incentivize miners to do so.

The question is whether the prescribed strategy is a Nash-equilibrium~\cite{aumann1987correlated,friedman1971nonCooperative,van2002strategic}, meaning no miner can benefit from individually deviating to a different strategy.
Like in previous work~\cite{nakamoto2008bitcoin,eyal2014majority,sapirshtein2016optimal,nayak2016stubborn,pass2017fruitchains}, the~$\protocolPropertyAdversarial$ metric is the maximal relative miner balance that achieves this: If all miners have relative balances smaller than the threshold, then the prescribed strategy is a Nash-equilibrium.

Formally, denote by~$\strategyNameProtocol{i,\bestResponseStrategyName}{\protocolGenericName}$ the best-response strategy of miner~$i$ with relative balance~$\balanceNormalizedSub{i}$ when all other miners follow~$\strategyNameProtocol{\defaultStrategyName}{\protocolGenericName}$.
$\protocolPropertyAdversarial$ is the maximal value~$\balanceNormalizedSub{i}$ such that~$\strategyNameProtocol{i,\bestResponseStrategyName}{\protocolGenericName} = \strategyNameProtocol{\defaultStrategyName}{\protocolGenericName}$.
It follows that~$\strategyNameProtocol{\defaultStrategyName}{\protocolGenericName}$ is a Nash-equilibrium if all miner relative balances are not greater than~$\protocolPropertyAdversarial$.

		\subsubsection*{\textbf{Nakamoto}}

Sapirshtein et al.~\cite{sapirshtein2016optimal} showed that for~$\protocolNameBitcoin$ with the uniform tie-breaking fork selection rule~ (see \S\ref{sec:bitcoin_protocol}) the metric value is~$\protocolPropertyAdversarial = 0.232$.

		\subsubsection*{\textbf{HEB}}

An optimal miner strategy must consider how to allocate the balance, which previous blocks to point to, what block type to create, and when to publish created blocks.

\paragraph*{\textbf{PoW-only analysis}}

Before the general analysis we start by considering a specific, natural,~\emph{PoW-only} strategy~$\strategyNameProtocol{\ignorantMiningStrategyName}{\hebShortName}$~(Alg.~\ref{alg:ignorant_mining_strategy}), which simply ignores the internal expenditure aspect of~$\hebShortName$.

The idea of~$\strategyNameProtocol{\ignorantMiningStrategyName}{\hebShortName}$ is to maximize the block creation rate by expending all resources externally.
The miner tries to create all the epoch blocks herself, and thus obtain all the epoch rewards.
Specifically, miner~$i$ allocates her balance s.t.~$\left\langle 0,\balanceSubSuper{i}{} \right\rangle$, creates regular blocks pointing to the last block in her local storage~$\localStorage{i}$, and publishes these blocks only after she has created~$\epochLength$ of them.

$\strategyNameProtocol{\ignorantMiningStrategyName}{\hebShortName}$ does not allocate resources for factored blocks, as if successful, all published blocks on the main chain are by the miner, thus her relative block weight is 1, granting her all the epoch reward.

This strategy is of interest as it abuses the internal expenditure mechanism; it is also simple enough to lend itself to a closed-form analysis.
Specifically, for~$\strategyNameProtocol{\ignorantMiningStrategyName}{\hebShortName}$ to be expected to succeed, miner~$i$ has to create blocks at a rate higher than all other miners combined.

Let~$\balanceSubSuper{\neg i}{}$ be the aggregate balance of all miners except miner~$i$.
It follows the requirement is~$\balanceSubSuper{i}{\externalCurrency} > \balanceSubSuper{\neg i}{\externalCurrency}$, and we seek the minimal~$\balanceNormalizedSubSuper{i}{}$ satisfying this condition.

All miners except miner~$i$ follow~$\strategyNameProtocol{\defaultStrategyName}{\hebShortName}$, so~$\balanceSubSuper{\neg i}{\externalCurrency} = \oneMinusPosRatioBrackets \balanceSubSuper{\neg i}{}$.
We get the aforementioned inequality is equivalent to~$\balanceSubSuper{i}{} > \oneMinusPosRatioBrackets \balanceSubSuper{\neg i}{}$.
Now, recall that~$\balanceSubSuper{i}{} + \balanceSubSuper{\neg i}{} = \epochLength$ and~$\epochLength \balanceNormalizedSubSuper{i}{} = \balanceSubSuper{i}{}$, and we get~$\balanceNormalizedSubSuper{i}{} > \oneMinusPosRatioBrackets \left(1 - \balanceNormalizedSubSuper{i}{}\right)$ or~$\balanceNormalizedSubSuper{i}{} > \tfrac{1 - \posRatio}{2 - \posRatio}$.

The value of~$\tfrac{1 - \posRatio}{2 - \posRatio}$ is an upper bound of~$\protocolPropertyAdversarial $; for larger~$\balanceNormalizedSubSuper{i}{}$ values, the strategy~$\strategyNameProtocol{\defaultStrategyName}{\hebShortName}$ is not the best-response, i.e., not a Nash-equilibrium.

We present this bound in Fig.~\ref{fig:taking_over_the_system}.
If~$\balanceNormalizedSub{i} > \tfrac{1 - \posRatio}{2 - \posRatio}$, i.e., above the plot,~$\strategyNameProtocol{i,\ignorantMiningStrategyName}{\hebShortName}$ outperforms~$\strategyNameProtocol{\defaultStrategyName}{\hebShortName}$, showing that~$\strategyNameProtocol{\defaultStrategyName}{\hebShortName}$ is not a Nash-equilibrium.

\begin{figure}[!t]
	
	\centering
	\resizebox{\VERTICALFIGSCALEINPUT\textwidth}{!}{	{\renewcommand\normalsize{\LARGE}\normalsize			\input{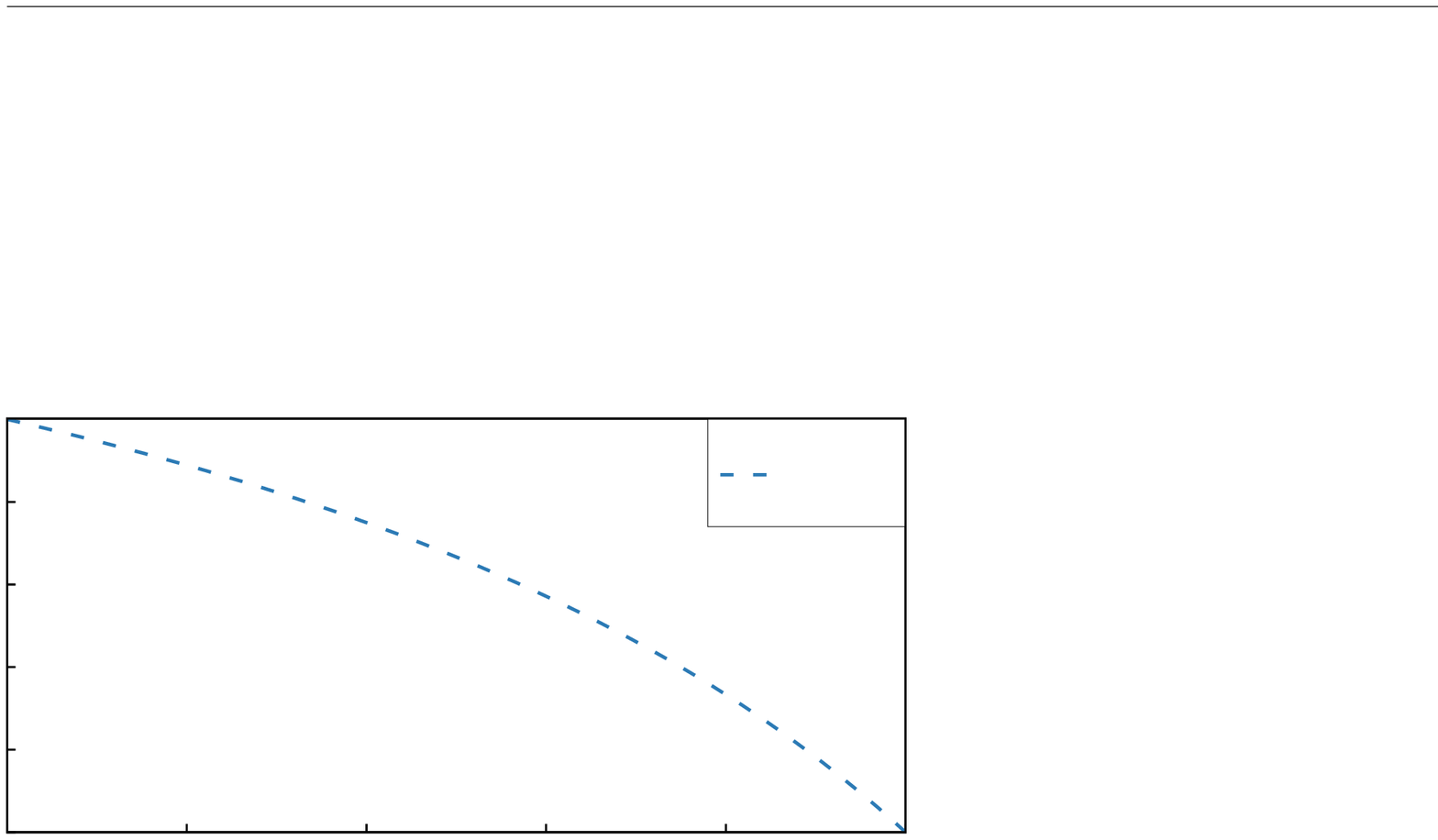}}}	
	\vspace{-1.2\baselineskip}
	\caption{Required $\balanceNormalizedSub{i}$ for~$\strategyNameProtocol{i,\ignorantMiningStrategyName}{\hebShortName}$ to outperform~$\strategyNameProtocol{i,\defaultStrategyName}{\hebShortName}$.}
	\label{fig:taking_over_the_system}
	\vspace{-1.2\baselineskip}	
\end{figure}

As expected, higher~$\posRatio$ values lower the bound, as miner~$i$ is competing against less external balance.
This result matches~$\protocolNameBitcoin$, as if~$\posRatio = 0$ then~$\tfrac{1 -\posRatio}{2 -\posRatio} = 0.5$, yielding the established 50\% bound~\cite{nakamoto2008bitcoin,Andes2011kryptonite,sapirshtein2016optimal,kroll2013economics}.


\paragraph*{\textbf{General analysis}}

Following previous work~\cite{sapirshtein2016optimal,gervais2016security,zur2020efficient,hou2019squirrl}, we use \emph{Markov Decision Process} (\emph{MDP}) to search for the optimal strategy in~$\hebShortName$.
The MDP includes the internal expenditure and block weights, and produces miner~$i$'s best-response strategy~$\strategyNameProtocol{i,\bestResponseStrategyName}{\hebShortName}$ based on system parameters.


We note the state and action spaces grow exponentially with the epoch length, limiting available analysis to relatively small epoch values.
Therefore, similarly to previous work~\cite{sapirshtein2016optimal,gervais2016security,zur2020efficient,hou2019squirrl}, we also limit the state space by excluding strategies requiring longer, and thus less probable, sequences of events.

Our focus is finding the required parameter values for which the best response~$\strategyNameProtocol{i,\bestResponseStrategyName}{\hebShortName}$ is the prescribed~$\strategyNameProtocol{i,\defaultStrategyName}{\hebShortName}$ strategy.
Recall that~$\strategyNameProtocol{i,\bestResponseStrategyName}{\hebShortName}$ is the optimal implementation of
$\funcNameAllocateBudgtWithFont{}$,~$\funcNameGenerateBlockWithFont$ and~$\funcNamePublishWithFont$ given~$\balanceSub{i}$ and the system parameters~$\epochLength,\factor, \posRatio$, hence we take the following approach.

We fix~$\epochLength = 10$ to limit the state space, and for various values of~$\posRatio$ and~$\balanceNormalizedSub{i}$ we use binary-search to find the minimal~$\factor \in \left[1,10^8\right]$ value such that~$\strategyNameProtocol{i,\bestResponseStrategyName}{\hebShortName} = \strategyNameProtocol{\defaultStrategyName}{\hebShortName}$.
First, we consider~$\funcNameAllocateBudgtWithFont{}$ implementations that let miner~$i$ create a natural number of blocks (allocating balance to enable the creation of a fraction of a block is strictly dominated, enabling discretization of possible implementations).
For each such implementation we use the MDP to obtain the optimal implementation of~$\funcNameGenerateBlockWithFont$ and~$\funcNamePublishWithFont$.
We let the miner play the resultant strategies for~$5000$ games each, and take the most rewarding to be~$\strategyNameProtocol{i,\bestResponseStrategyName}{\hebShortName}$.

\begin{figure}[!t]
	\vspace{-0.5\baselineskip}
	\centering
	\resizebox{\VERTICALFIGSCALEINPUT\textwidth}{!}{	{\renewcommand\normalsize{\LARGE}\normalsize			\input{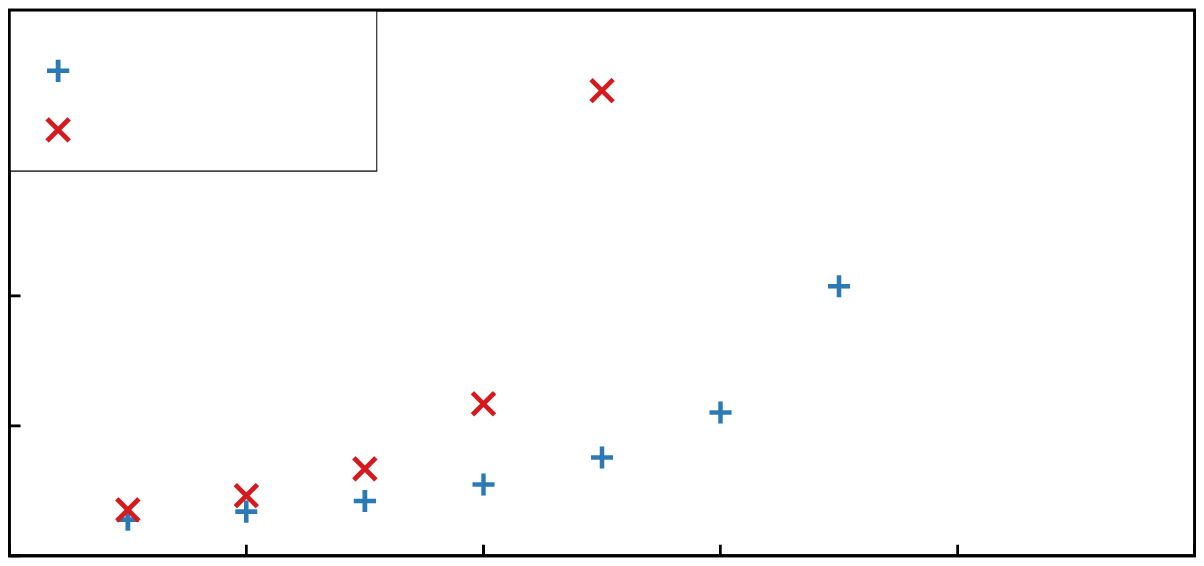}}}	
	
	\vspace{-1.3\baselineskip}
	\caption{Minimal required~$\factor$ for~$\protocolPropertyAdversarial$.}
	\label{fig:mdp_results}
	\vspace{-1.4\baselineskip}
\end{figure}

We present the results in Fig.~\ref{fig:mdp_results}, showing that increasing~$\factor$ values and lowering~$\posRatio$ increases~$\protocolPropertyAdversarial$.
Specifically, for~$\balanceNormalizedSub{i}=0.2$ the required~$\factor$ values grow exponentially with~$\posRatio$ up to~$\posRatio = 0.5$, and from there even the maximal~$\factor$ value does not accommodate the desired behavior.
We note a similar behavior for~$\balanceNormalizedSub{i}=0.1$, growing exponentially with~$\posRatio$ up to~$\posRatio = 0.7$, being the maximal~$\posRatio$ that leads to~$\strategyNameProtocol{\defaultStrategyName}{\hebShortName}$ being a Nash-equilibrium.

We also note that lower~$\balanceNormalizedSubSuper{i}{}$ requires lower~$\factor$ values, and specifically, 
there are no~$\factor$ and~$\posRatio$ values for which the configuration of~$\balanceNormalizedSubSuper{i}{} = 0.3$ achieves a Nash-equilibrium.
This is expected as the profitability threshold for selfish-mining variants is~$\balanceNormalizedSubSuper{i}{} = 0.232$~\cite{sapirshtein2016optimal}, and indeed the resultant best-response strategies resemble selfish-mining in~$\protocolNameBitcoin$.

We conclude that~$\protocolPropertyAdversarial$ relies on~$\epochLength$,~$\factor$ and~$\posRatio$; by setting~$\factor=20$, we can obtain~$\protocolPropertyAdversarial = 0.2$ even for~$\posRatio = 0.5$, close to $\protocolNameBitcoin$'s $\protocolPropertyAdversarial=0.232$ value~\cite{sapirshtein2016optimal}.

\paragraph*{\textbf{MDP Technical Details}}

Like previous work~\cite{sapirshtein2016optimal,gervais2016security,zhang2019lay,hou2019squirrl,zur2020efficient}, our MDP finds the best-response strategy of miner~$i$ for given system parameters.
We detail the modeling of the other system miners as a single one, denoted miner~$\neg i$, who follows a prescribed strategy.

As such, we limit the analysis to strategies considering at most two chains at any given time~\cite{arvindcutoff,eyal2014majority,sapirshtein2016optimal,nayak2016stubborn}.
The first is the~\emph{public} chain, followed by miner~$\neg i$, while the other is known only to miner~$i$ (i.e., maintained in her local storage~$\localStorage{i}$), named the~\emph{secret} chain.

We describe~miner~$\neg i$, the action space, the state space, and the reward function.

\paragraph{Miner~$\neg i$}
As in previous work~\cite{eyal2014majority,arvindcutoff,sapirshtein2016optimal}, miner~$\neg i$ is a cohort comprising infinitely-many, non-colluding infinitely-small balance miners.
She follows a~\emph{petty-compliant}~\cite{arvindcutoff} strategy~$\strategyNameProtocol{\pettyCompliantStrategyName}{\hebShortName}$ that is a variant of~$\strategyNameProtocol{\defaultStrategyName}{\hebShortName}$~(Alg.~\ref{alg:strategy_petty_compliant}): it tie-breaks conflicting longest chains uniformly-at-random from the multiple longest chains~\emph{with the minimum accumulated weight}.


The intuition for~$\strategyNameProtocol{\pettyCompliantStrategyName}{\hebShortName}$ is as follows.
Given a specific block, its relative weight is higher on a chain with less accumulated weight.
As miners get reward based on their relative weight, they have an explicit incentive to pick the chain that has a lower accumulated weight.
That is, this strategy is more logical from a miner's perspective as it is expected to increase her utility.

So, assuming a petty-compliant cohort strengthens our result, producing a more conservative threshold, as such behavior is more amendable to manipulation~\cite{arvindcutoff}.

\begin{figure}[!t]
	
	\makeatletter
	\newcommand{\removelatexerror}{\let\@latex@error\@gobble}
	\makeatother
	
	\removelatexerror

	\begin{minipage}[t]{.27\textwidth}
		\vspace*{-\baselineskip}
		\begin{algorithm}[H]
			\SetNoFillComment
			\SetAlgoNoLine 
			\SetAlgoNoEnd 
			\DontPrintSemicolon 
			\caption{$\strategyNameProtocol{i,\pettyCompliantStrategyName}{\hebShortName}$} 
			\label{alg:strategy_petty_compliant}

			\Fn{\FAllocateBudgt{\funcArgsAllocateBudgt{i}}}{
				\Return~$\oneMinusPosRatioBrackets \balanceSub{i}$\;
			}
			\Fn{\FGenerateBlock{\funcArgsGenerateBlock{i}}}{
				$\mainChainName{} \gets$  uniformly from the minimal-weight chains in~$\funcLongestChains{}$
				
				$\text{pointer} \gets \lastBlockOnMainChain{\mainChainName{}}$\;
				\If{
					$\countFactoredBlocksWithChain{i}{\mainChainName{}} < \floor*{\dfrac{\balanceSubSuper{i}{\internalCurrency}}{\posRatio \rewardPerBlock}}$}
				{~$\text{type} \gets \textit{factored}$}
				\Else{~$\text{type} \gets \textit{regular}$}
				\Return NewBlock(pointer, type)\;
			}
			\Fn{\FPublish{\funcArgsPublish{i}}}{
				\Return All previously unpublished blocks\;
			}
		\end{algorithm}
	\end{minipage}%
	\hfill
	\begin{minipage}[t]{.20\textwidth}
		\vspace*{-\baselineskip}
		\begin{algorithm}[H]
			\SetNoFillComment
			\SetAlgoNoLine 
			\SetAlgoNoEnd 
			\DontPrintSemicolon 
			\caption{$\strategyNameProtocol{i,\ignorantMiningStrategyName}{\hebShortName}$} 
			\label{alg:ignorant_mining_strategy}

			\Fn{\FAllocateBudgt{\funcArgsAllocateBudgt{i}}}{
				\Return~$\left\langle 0,\balanceSub{i} \right\rangle~$\;
			}

			\Fn{\FGenerateBlock{\funcArgsGenerateBlock{i}}}{
				$\text{pointer} \gets \lastBlockOnMainChain{\localStorage{i}} $\;
				\Return NewBlock(pointer,\textit{regular})\;
			}

			\Fn{\FPublish{\funcArgsPublish{i}}}{
				
				\If{
					$\countFactoredBlocksWithChain{i}{\localStorage{i}} < \epochLength$}
				{\Return \;}
				\Else{\Return All previously unpublished~$\epochLength$ blocks\;}
				
			}
		\BlankLine
		\BlankLine		
		\BlankLine
		\BlankLine			
		\BlankLine
		\BlankLine				
		\end{algorithm}
	\end{minipage}
	
	\vspace{-1.8\baselineskip}
	
\end{figure}

In practice it is less likely that a miner will prefer to extend a chain excluding a previous block she created, as that lowers her block weight and utility.
So, larger honest miners are more reluctant to discard their blocks and follow other chains; modeling the cohort to comprise infinitely-many negligibly-small balance miners removes this consideration~\cite{arvindcutoff}, resulting with a stronger adversary, i.e., producing a more conservative bound.
We illustrate that with the following example.

\begin{myExample}
	
	Assume the last block~$\lastBlockOnMainChain{\mainChainName{}}$ on the main chain~$\mainChainName{}$ is a factored block created by miner~$\neg i$.
	According to~$\strategyNameProtocol{\defaultStrategyName}{\hebShortName}$ the next block should point to~$\lastBlockOnMainChain{\mainChainName{}}$ and extend~$\mainChainName{}$.
	Assume the scheduler lets miner~$i$ create the next block, and she creates a regular block that points to the same block as~$\lastBlockOnMainChain{\mainChainName{}}$, publishing it immediately.
	Denote the original and the new chain by~$\mainChainName{1}$ and~$\mainChainName{2}$, respectively.
	
	Chains~$\mainChainName{1}$ and~$\mainChainName{2}$ are the longest chains, i.e.,~$\funcLongestChains{} = \left\{\mainChainName{1},\mainChainName{2}\right\}$.
	Now, assume that miner~$i$ points her next created block to~$\lastBlockOnMainChain{\mainChainName{2}}$.
	As~$\weightWithChain{}{\mainChainName{1}} - \weightWithChain{}{\mainChainName{2}} = \factor - 1 > 0$ and miner~$\neg i$ follows~$\strategyNameProtocol{\pettyCompliantStrategyName}{\hebShortName}$, if miner~$\neg i$ is picked to create the next block she will deterministically choose to point it to~$\lastBlockOnMainChain{\mainChainName{2}}$.
	
	Therefore, regardless of which miner gets to create the next block,~$\lastBlockOnMainChain{\mainChainName{1}}$ will not be pointed by following blocks, effectively removing it from any future longest chain.
	That means miner~$i$ had managed to replace a factored block miner~$\neg i$ with her own regular block on the main chain, increasing her block weight while decreasing that of miner~$\neg i$, both effectively increasing her utility.
	
	If miner~$\neg i$ had followed~$\strategyNameProtocol{\defaultStrategyName}{\hebShortName}$ then she would have pointed her next block to~$\lastBlockOnMainChain{\mainChainName{2}}$ only with probability~$0.5$, as both~$\mainChainName{1}$ and are~$\mainChainName{2}$ are of the same length.
	That means that block~$\lastBlockOnMainChain{\mainChainName{1}}$ might still end on the main chain (depending on miner~$i$'s strategy and future block creations), resulting with lower expected block weight for miner~$i$.

\end{myExample}

\paragraph{Action space}
We represent miner~$i$'s action as a two-element tuple~$\actionContent$.
Element~$\actionFieldOneName$ describes how miner~$i$ interacts with the secret and public chains, and may contain one of the three values~---~\emph{publish},~\emph{adopt} and~\emph{wait}.
The value of~\emph{publish} indicates the miner publishes the blocks of the secret chain, a value of~\emph{adopt} indicates the miner abandons the secret chain and adopts the public chain, and~\emph{wait} indicates the miner does neither the former nor the latter.
Element~$\actionFieldTwoName$ takes a binary value, describing whether the next block the miner creates is factored.

\paragraph{State space}
We represent states as a three-element tuple~$\nodeContent$~\cite{sapirshtein2016optimal,zur2020efficient,gervais2016security,hou2019squirrl}.
Elements~$\fieldOneName$ and~$\fieldTwoName$ represent the contents of the secret and public chains, respectively.
Element~$\fieldThreeName$ has a binary value indicating whether miner~$\neg i$ is partitioned with regards to which of the two chains to extend.
Note that~$\fieldThreeName$ is true only if miner~$i$ had previously published her chain.

Previous work~\cite{sapirshtein2016optimal,gervais2016security,zur2020efficient} analyzes an infinite game and introduces a~\emph{truncation parameter}~$T$, capping the length of the secret and public chains.
The used state space includes counters of the blocks in the secret and public chains, resulting in a state space complexity of~$\mathcal{O}\left(T^2\right)$.

This state space does not fit~$\hebShortName$~-- we must maintain the order of block types on each chain, resulting in a state space complexity of~$\mathcal{O}\left(2^{\epochLength}\right)$.
As an example, consider the case where miner~$i$ has first created a factored block and then a regular block on the secret chain, followed by miner~$\neg i$ creating a factored block on the public chain.
In this case, miner~$i$ publishing the first factored block results in a fork, resembling~\emph{Lead-Stubborn Mining}~\cite{nayak2016stubborn}.
Now, consider the similar case where miner~$i$ has first created a regular block and then a factored block on the secret chain, followed by miner~$\neg i$ creating a factored block on the public chain.
In this case, miner~$i$ publishing the first regular block results with miner~$\neg i$ deterministically adopting it and forfeiting her recently-created factored block.
To distinguish these two cases, our state space includes the order of created block types.

\paragraph{Reward function}
Rewarding states are those where either the secret or the public chain are of~$\epochLength$ blocks, that is,~$\lengthMainChain{\fieldOneName} = \epochLength$ or~$\lengthMainChain{\fieldTwoName} = \epochLength$.
Note that this restricts miner~$i$ to strategies bounded by the creation of~$\epochLength$ blocks; strategies that exceed this limitation are feasible with negligible probability~\cite{gervais2016security}.
These states indicate the epoch conclusion and hence the reward distribution of~$\hebShortName$.

\begin{myNote}
	
	Previous work~\cite{sapirshtein2016optimal,gervais2016security,zur2020efficient} analyzed infinitely-repeating games, hence their MDPs had no final states.
	They iteratively optimize for the best-response strategy, stopping only when meeting a predefined precision criteria.
	This results with an approximation of the best-response strategy.
	
	The~$\hebShortName$ MDP has final states (where the longest chain is of length~$\epochLength$), and there are no state reoccurrences.
	Essentially, the~$\hebShortName$ MDP is a dynamically-programmed search on all possible strategies.
	As such, its result is not an approximation but the actual best-response strategy. 
	
\end{myNote}

		\subsection{$\protocolPropertyDoubleSpendAttack$ and $\protocolPropertySabotageAttack$}
		\label{sec:protocolPropertyDoubleSpendAttack}

We consider safety-violation attacks~\cite{rosenfeld2014hashrateAnalysis,karame2012fast,bonneau2016buy,mccorry2018smart,judmayer2019pay,conti2018survey,bano2019sok,redman2020bitcoingold51,bitcoin2020etc51,dcimit202051attacks} as scenarios where an attacker causes the system to make an invalid transition.
This can be achieved by creating and publishing an alternative chain, surpassing in length the main one.
The blocks of the original chain are then discarded, and the system state is reinstated according to the blocks on the alternative, new chain.

To mount this attack in~$\protocolNameBitcoin$ the attacker expends her resources on creating blocks to form the alternative chain; recall that each block costs its worth in reward to create (Eq.~\ref{eq:model_total_utilities_equals_expenses}).
Therefore, if the attack is successful, the attacker is fully compensated for her expenditures by the rewards from her created blocks.
As such, there is a threshold of required resources to mount this attack, but once met, the attack is \emph{free}.

The metric~$\protocolPropertyDoubleSpendAttack$ measures the minimal required balance for a miner to deploy such a refunded safety-violation attack on the system, assuming all other miners follow the prescribed strategy~$\strategyNameProtocol{\defaultStrategyName}{\protocolGenericName}$.
As shown in previous work~\cite{bonneau2016buy}, the attacker may rent vast computational resources for a short period of time or a moderate amount for longer periods.
We therefore measure the expected cost to create a single block, disregarding the attack duration and amplitude.

The~$\protocolPropertySabotageAttack$ metric removes the refund requirement, and simply represents the cost to create a block.

Formally, assume all miners follow~$\strategyNameProtocol{\defaultStrategyName}{\protocolGenericName}$.
Then, $\protocolPropertyDoubleSpendAttack$ is the minimal cost to create a block, guaranteeing full compensation should it be on the main chain, and $\protocolPropertySabotageAttack$ is this cost without any further compensation guarantees.

\subsubsection*{\textbf{Nakamoto}}

In~$\protocolNameBitcoin$ the cost to create each block is~$1$, hence~$\protocolPropertyDoubleSpendAttack = 1$.
All blocks produce the same reward, hence a miner cannot reduce the cost for a safety-violation attack by choosing to create less-rewarding blocks.
Therefore,~$\protocolPropertySabotageAttack = 1$.

\subsubsection*{\textbf{HEB}}

In equilibria the total external expenses are~$1 -\posRatio$ of the total balances, that is~$\balanceOfMinersExternal = \oneMinusPosRatioBrackets \balanceOfMinersAll$.
As~$\balanceOfMinersAll = \epochLength$ it follows that the required external expenses to create a single block is~$1 -\posRatio$.

As other miners create factored blocks, a miner also has to create a factored block to be fully compensated for her expenses, requiring additional spending of~$\posRatio$.
Hence, the cost to create a single block is~$1$, so~$\protocolPropertyDoubleSpendAttack = 1$.
That is,~$\hebShortName$ is as resilient to refunded attacks as~$\protocolNameBitcoin$.

Alternatively, a miner can disregard compensation and choose to create regular blocks, baring no additional internal expenses, and so~$\protocolPropertySabotageAttack =~1~-~\posRatio$, which is less secure than~$\protocolNameBitcoin$.
However, the lack of direct compensation makes these attacks very expensive, hence they are only available to a well-funded adversary with an exogenous utility, e.g., interested in destabilizing or short-selling a cryptocurrency.

Indeed, previous attack instances~\cite{redman2020bitcoingold51,bitcoin2020etc51,dcimit202051attacks,voell2020etc51} were on relatively-small systems and were of the former, refunded type.
We are not aware of such sabotage attacks happening in practice; this is possibly because the required expenditure surpasses the potential profit~\cite{conti2018survey,bano2019sok}.

	\subsection{$\protocolPropertyPermissiveness$}
	\label{sec:protocolPropertyPermissiveness}

Cryptocurrency protocols implement their own reward distribution mechanisms~\cite{chen2019axiomaticApproach}, and may choose to condition rewards on a miner having the internal system currency~$\internalCurrency$.
For example, in PoS systems~\cite{kiayias2017ouroboros,gilad2017algorand,eosio} owning~$\internalCurrency$ is a requisite, and miners without~$\internalCurrency$ cannot participate and get rewards.
In contrast, in PoW systems~\cite{nakamoto2008bitcoin,buterin2013ethereum} owning~$\internalCurrency$ does not affect reward eligibility.

Acquisition of~$\internalCurrency$ involves an update of the new currency ownership in the system state.
This requires the cooperation of the present system miners: 
They decide which state updates occur when placing user data in their created blocks.
So, if token ownership is a mining requirement, then a new miner wishing to participate requires the cooperation of existing miners.

Previous work considered either permissioned systems that require token ownership~\cite{gilad2017algorand,kiayias2017ouroboros,ef2020eth2Pos,goodman2014tezos} (some also require explicitly locking owned tokens as a collateral), or permissionless systems~\cite{nakamoto2008bitcoin,buterin2013ethereum,eyal2016bitcoin} that do not.

We generalize this binary differentiation to a continuous metric,~$\protocolPropertyPermissiveness$, measuring the revenue of a newly-joining miner without cooperation from the incumbents.
The metric is the ratio between the revenues of a miner where she failed or managed to obtain~$\internalCurrency$.

Formally, consider a miner~$i$ with balance~$\balanceSubSuper{i}{}$, and assume that all other miners follow~$\strategyNameProtocol{\defaultStrategyName}{\protocolGenericName}$.
Denote by~$\strategyNameProtocol{\noIcStrategyName}{\protocolGenericName}$ a strategy identical to~$\strategyNameProtocol{\defaultStrategyName}{\protocolGenericName}$ with the exception that the $\funcAllocateBudgt{i}$ implementation returns~$\left\langle 0,\balanceSubSuper{i}{} \right\rangle$.
Note this captures the inability of miner~$i$ to obtain~$\internalCurrency$.
Denote by~$\utilityNoIcfuncName{i}{\protocolGenericName}$ and by~$\utilityYesIcfuncName{i}{\protocolGenericName}$ the utility of miner~$i$ if she follows~$\strategyNameProtocol{\noIcStrategyName}{\protocolGenericName}$ and~$\strategyNameProtocol{\defaultStrategyName}{\protocolGenericName}$, respectively.
We then define $\protocolPropertyPermissiveness\triangleq\tfrac{\utilityNoIcfuncName{i}{\protocolGenericName}}{\utilityYesIcfuncName{i}{\protocolGenericName}}$.

If~$\protocolPropertyPermissiveness=1$ then a miner's utility is not affected by her inability to obtain~$\internalCurrency$, meaning the protocol is permissionless.
In contrast,~$\protocolPropertyPermissiveness=0$ indicates that a miner who cannot obtain~$\internalCurrency$ is completely prevented from participation.

\subsubsection*{\textbf{Nakamoto}}

As a pure PoW blockchain protocol, $\protocolNameBitcoin$ does not require~$\internalCurrency$ balance, and both strategies do not allocate any balance internally, and as such~$\protocolPropertyPermissiveness=1$.

\subsubsection*{\textbf{HEB}}

Deriving~$\utilityNoIcfuncName{i}{\protocolGenericName}$ and~$\utilityYesIcfuncName{i}{\protocolGenericName}$ leads to~$\protocolPropertyPermissiveness = \tfrac{1}{\balanceNormalizedSubSuper{i}{} + \factor \left(1 - \balanceNormalizedSubSuper{i}{}\right)}$~(Appendix~\ref{app:permissiveness}).

\begin{figure}[!t]
	\vspace{-0.5\baselineskip}	
	\centering
	\resizebox{\VERTICALFIGSCALEINPUT\textwidth}{!}{	{\renewcommand\normalsize{\LARGE}\normalsize			\input{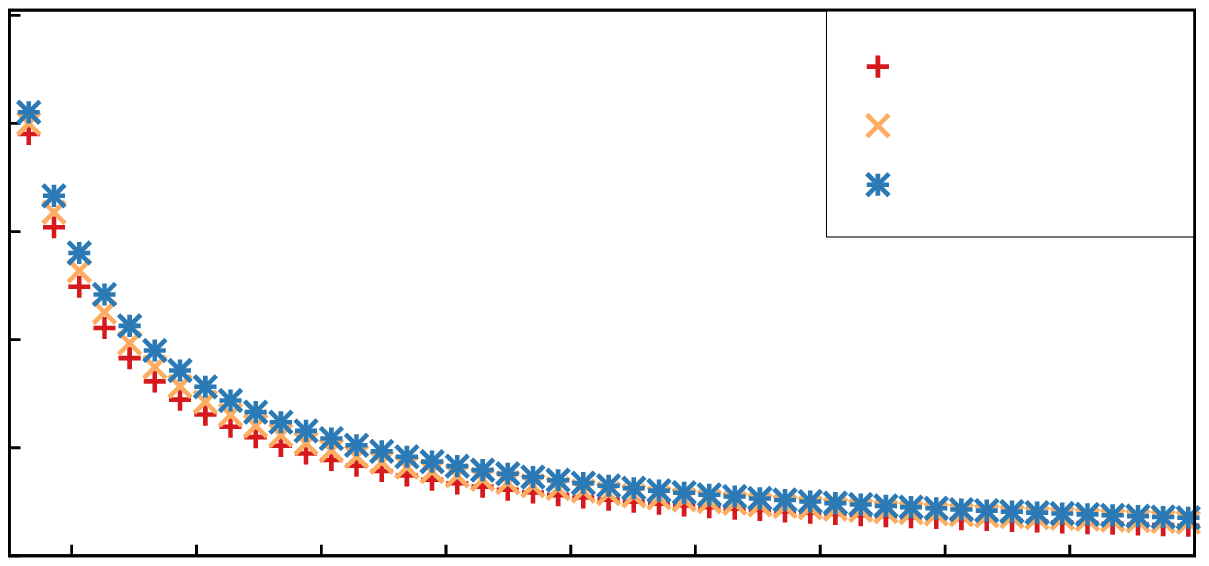}}}	
	\vspace{-1.2\baselineskip}	
	\caption{$\hebShortName$~$\protocolPropertyPermissiveness$.}
	\label{fig:censor_resilience}
	\vspace{-1.7\baselineskip}	
\end{figure}

Fig.~\ref{fig:censor_resilience} presents~$\protocolPropertyPermissiveness$ for different values of~$\balanceNormalizedSub{i}$ as a function of~$\factor$.
It shows that higher factor values~$\factor$ lead to lower~$\protocolPropertyPermissiveness$ values, making the system more permissioned.
It also shows that miners with higher relative balances are slightly less susceptible to these effects.

As expected, higher values of~$\factor$ increase the reward of creating factored blocks, and failing to do so results in lower income.
Additionally, a higher relative balance enables creating more blocks, decreasing the overall block weight of the other miners, making the miner's revenue less susceptible.

Although failure to obtain~$\internalCurrency$ results with a lower reward, it still enables the new miner to create blocks herself, removing the requirement for cooperation from the incumbents in the subsequent epochs.
The reduced reward in the first epoch is a one-time cost that is negligible for a long-running miner.

This is a significant and qualitative improvement over permissioned systems, where a miner that cannot obtain tokens~\cite{gilad2017algorand} or lock them as a collateral~\cite{kiayias2017ouroboros,goodman2014tezos,ef2020eth2Pos} is blocked from all future participation.

	\subsection{$\protocolPropertyPow$}
	\label{sec:protocolPropertyPow}		

$\protocolPropertyPow$ evaluates the external expenditure of the protocol, and lower values indicate a lower environmental impact.
Formally, assume all miners follow~$\strategyNameProtocol{\defaultStrategyName}{\protocolGenericName}$.
$\protocolPropertyPow$ is the total of miner external expenses, measured in~$\externalCurrency$, normalized by the epoch length, i.e,~$\protocolPropertyPow\triangleq\tfrac{\balanceOfMinersExternal}{\epochLength}$.

\subsubsection*{\textbf{Nakamoto}}

The total miner expenses are~$\balanceOfMinersExternal=\epochLength$, so~$\protocolPropertyPow=1$.

\subsubsection*{\textbf{HEB}}

When all miners follow~$\strategyNameProtocol{\defaultStrategyName}{\hebShortName}$ then~$\balanceOfMinersExternal = \oneMinusPosRatioBrackets \balanceOfMinersAll$ and~$\protocolPropertyPow = 1 - \posRatio$.
This is the main advantage of~$\hebShortName$ over~$\protocolNameBitcoin$.

\section{Practical Parameters}
\label{sec:concrete_instantiation_parameters}

As we have seen, $\hebShortName$ presents several knobs for the system designer.
Longer epoch length~$\epochLength$ improves~$\protocolPropertyDecentralization$, however, also means that reward distribution takes longer.
Higher factored block weight~$\factor$ improves~$\protocolPropertyAdversarial$ at the expense of~$\protocolPropertyPermissiveness$.
Higher internal expenditure rate~$\posRatio$ reduces the external expenditures, but makes the system less robust against rational miners, and reduces the required costs for sabotage attacks.

The choice of parameter values should be according to the desired system properties.
Each system has different goals, and we emphasize that determining optimal parameter values is not a goal of this work.
Nevertheless, in this section we consider a specific parameter choice.
We compare this instantiation to Bitcoin, and use the latter's miner balance distribution~\cite{blockchain2020hashrateDistribution} as a representative example.

We choose the external cost parameter to be~$\posRatio = 0.5$, the epoch length to be~$\epochLength = 1000$, and the factor to be~$\factor = 20$.

First and foremost, this setting results with only half of the external resource consumption ($\protocolPropertyPow=0.5$), which is equivalent to reducing the entire power consumption of Denmark~\cite{de2020bitcoin,countryEnergyConsumption2020citypopulation}.
This choice incentivizes rational miners with up to~$0.2$ relative balance to follow the prescribed strategy ($\protocolPropertyAdversarial=0.2$) down from Bitcoin's~$0.232$~\cite{sapirshtein2016optimal}.
Note that the largest miner in Bitcoin has a relative balance of~$0.16$~\cite{blockchain2020hashrateDistribution}, so rational miners would follow the prescribed, honest mining behavior.

With Bitcoin's expected block creation interval of 10 minutes, having epochs of~$\epochLength = 1000$ means mining rewards are distributed on a weekly basis.
This is longer than the seventeen hours~\cite{nakamoto2008bitcoin,btcWiki2020blockchain} miners wait today, but arguably still an acceptable time frame.

The threshold for a refunded safety-violation ($\protocolPropertyDoubleSpendAttack = 1$), is as in Bitcoin, but the non-refunded variation is twice as cheap ($\protocolPropertySabotageAttack = 0.5$).
We note that we are not aware of attacks of either type on prominent cryptocurrency systems, and that the non-refunded type is unlikely due to the lack of endogenous compensation~(cf.~\S\ref{sec:protocolPropertyDoubleSpendAttack}).

In regards to permissiveness, a miner with a relative balance~$\balanceNormalizedSubSuper{i}{}=0.1$ that fails to obtain any~$\internalCurrency$ due to incumbents is expected to get~$5\%$ of what she would have if she had~$\internalCurrency$ ($\protocolPropertyPermissiveness = 0.05$).
Recall this is at most a negligible one-time cost~(cf.~\S\ref{sec:protocolPropertyPermissiveness}).

Finally, for the current Bitcoin miner balance distribution~\cite{blockchain2020hashrateDistribution} the maximal relative advantage from size differences is 0.1\% ($\protocolPropertyDecentralization = 0.001$).
We consider modifications to further decrease this value in Appendix~\ref{app:practical_considerations}.


\vspace{-0.4\baselineskip}
\section{Conclusion}

We propose a new PoW paradigm that utilizes internal expenditure as a balancing mechanism.
We present~$\hebShortName$~-- a generalization of Nakamoto's protocol that allows its designer to tune external resource expenditure. 
We link the values of internal and external currencies, and formalize evaluation metrics including a blockchain's resilience to sabotage and revenue-seeking attacks and permissiveness on a continuous scale. 
We explore the trade-offs in parameter choice.
We propose practical parameters based on Bitcoin's ecosystem that cut down by half the PoW expenditure (equivalent to reducing the power consumption of an entire country) while maintaining similar security guarantees against practical attacks.

Natural questions that arise from the discovery of~$\hebShortName$ are what should be the security target for cryptocurrency protocols, how to set the parameters dynamically, and how to govern them~\cite{goodman2014tezos,reijers2016governance}.
Beyond these,~$\hebShortName$ extends the design space of decentralized systems, and is a step forward in realizing secure PoW systems with a sustainable environmental impact.

%

\bibliographystyle{ACM-Reference-Format}

\bibliography{btc-new}
\appendix

	\section{PoW Alternatives}
	\label{app:related_pow_alternatives}

\paragraph{Permissioned systems}
These systems assume a set of entities that run a Byzantine fault tolerant~\cite{castro1999practical,kotla2007zyzzyva,kapitza2012cheapbft,gueta2018sbft,yin2019hotstuff,baudet2018state} algorithm to determine system state, and are also in charge of making membership changes.
Unlike~$\hebShortName$, these systems are~\emph{permissioned}, as a new participant requires the authorization of existing participants to join.
Such solutions are common in the enterprise market~\cite{cachin2016architecture,baudet2018state}, but are not directly suitable for permissionless cryptocurrency systems.

\paragraph{Proof of activity}
Bentov et al.~\cite{bentov2014proof} propose a system utilizing PoW to elect a committee of participants among the system stakeholders, who then run a variation of a classical permissioned distributed consensus protocol.
Similarly to Algorand~\cite{gilad2017algorand}, they also make assumptions regarding the availability of stake holders to ensure system progress.
$\hebShortName$ does not require assumptions of that sort since its blocks can be generated without tokens.

\paragraph{Proof of useful work}
This approach suggests performing useful work for PoW, that is, work that has external value outside the security of the cryptocurrency system~\cite{ball2017proofs,zhang2019lay,dotan2020proofs}.
However, these protocols also operate under different assumptions.
Ball et al.~\cite{ball2017proofs} relies on users providing problem instances altruistically, while Zhang et al.~\cite{zhang2017rem} rely on trusted hardware, i.e., a manufacturer like Intel. 

In this work we focus on reduction of wastefulness rather than its re-purposing, without relying on a centralized authority nor assuming altruistic behavior.

\paragraph{Proof of burn}
A different suggestion for PoW replacement is~\emph{proof-of-burn}~\cite{stewart2012proof,p4titan2014slimcoin,DBLP:journals/iacr/KarantiasKZ19}, in which miners prove the depletion of another cryptocurrency, typically PoW-based, to create blocks. 
This scheme replaces burning electricity with burning another currency.
This increases the scarcity of the burnt tokens, and hence that currency's value, maintaining high minting costs and negative environmental impact.
In contrast,~$\hebShortName$ reduces the expended external resources without deferring the waste.

\paragraph{Proof of space}
Permacoin~\cite{miller2014permacoin} and MeshCash~\cite{bentov2017tortoise} both use storage instead of computations as basis for their puzzles.
However, this approach still results with physical resource expenditure. 
$\hebShortName$ reduces any external expenditure, and can be applied to such systems as well.

\paragraph{PPCoin}

PPCoin~\cite{king2012peercoin} suggests a PoW system that transitions into PoS~-- the stake is measured by considering the time each token is held, creating a so-called \emph{coin-age}.
It requires similar assumptions to other PoS systems, and lacks evaluation of its security and incentive compatibility.
Specifically, it does not consider the reward (expressed either as coin or coin-age) nor the expenses (external expenses and coin-age) for creating PoW blocks, so the effect of using coin-age in the ecosystem is not well defined.
$\hebShortName$ uses the native system token for block generation, and has a full security and incentive-compatibility analysis.

\paragraph{Proof of transfer}
The Stacks blockchain~\cite{ali2020pox} uses~\emph{proof of transfer}~-- it requires transferring a base token (e.g., Bitcoin) to create blocks.
It does not contain a security proof nor incentive-compatibility analysis, preventing an accurate comparison.
Specifically, it does not take into account that a user can transfer her base tokens to herself, hence generating multiple transfer-proofs with only a minor investment.
$\hebShortName$ does not rely on a base token, the redistributed tokens are not transferred to a destination of choice but shared proportionally among all system entities, and the protocol has a full security and incentive-compatibility analysis.

	\section{Naive Protocols}
	\label{app:naive_protocols} 

We present three naive protocols and discuss their shortcomings.
These are simpler solutions compared to~$\hebShortName$, however, help clarify some of its design choices.

The first protocol, named~$\protocolNameBitcoinHalf$, is exactly as~$\protocolNameBitcoin$ excepts it mints new tokens at half rate~(\S\ref{app:naive_protocols_half}).
That is,~$\tfrac{\rewardPerBlock}{2}$ new tokens per block compared to~$\rewardPerBlock$ of~$\protocolNameBitcoin$.
We show this protocol is as wasteful as~$\protocolNameBitcoin$, i.e., they both result with~$\protocolPropertyPow = 1$.

We then consider a different protocol, named~$\protocolNameOneNoPrefixFullName$ ($\protocolNameOneNoPrefixShortName$), that similarly to~$\protocolNameBitcoin$ also mints~$\rewardPerBlock$ tokens per block~(\S\ref{app:naive_protocols_pr}).
However, for a parameter~$\posRatio \in \left[0,1\right)$, it rewards only~$\oneMinusPosRatioBrackets \rewardPerBlock$ of these to the block-generating miner, and distributes the other~$\posRatio \rewardPerBlock$ among all system miners and users proportionally to their current token holdings.
We show that~$\protocolNameOne$ is less wasteful than~$\protocolNameBitcoin$ ($\protocolPropertyPow = 1-\posRatio$ compared to~$1$) at the expense of less resiliency to free safety-violation attacks ($\protocolPropertyDoubleSpendAttack = 1-\posRatio$ compared to~$1$).

We conclude with the~$\protocolNameMandatoryFullName$ ($\protocolNameMandatory$) protocol, which is like~$\hebShortName$, but~\emph{requires} (rather than enabling) miners to internally-spend to create blocks~(\S\ref{sec:protocol_v3_strawman}).
Although providing lower~$\protocolPropertyDoubleSpendAttack$, this protocol undesirably provides~$\protocolPropertyPermissiveness=0$ on top of having an inherent liveness issue.

\subsection{$\protocolNameBitcoinHalf$}	
\label{app:naive_protocols_half} 

As mentioned,~$\protocolNameBitcoinHalf$ is just like~$\protocolNameBitcoin$ but mints tokens at a half rate.
The prescribed mining strategy is identical~(Alg.~\ref{alg:bitcoin_protocol_strategy}).

Recall that the real value of a single token in~$\externalCurrency$ is inverse to the number of circulating tokens (cf. \S\ref{sec:model_exchange_rate}).
As such, when minting at half rate, the price of a single token in~$\protocolNameBitcoinHalf$ is exactly twice of a token in~$\protocolNameBitcoin$.
So, despite a miner getting only~$\tfrac{\rewardPerBlock}{2}$ tokens with each block, her reward in~$\externalCurrency$ matches that of miners in~$\protocolNameBitcoin$.
As such,~$\protocolNameBitcoinHalf$ is just as wasteful as~$\protocolNameBitcoin$, both providing~$\protocolPropertyPow = 1$.

\subsection{$\protocolNameOneNoPrefixFullName$}	
\label{app:naive_protocols_pr} 

As previously stated,~$\protocolNameOne$ mints~$\rewardPerBlock$ tokens per block as~$\protocolNameBitcoin$, but distributes these minted tokens differently.
It gives~$\oneMinusPosRatioBrackets \rewardPerBlock$ to the miner who created the block, and the remaining~$\posRatio \rewardPerBlock$ are shared between the miners and users proportionally to their token holdings.
The prescribed mining strategy is as in~$\protocolNameBitcoin$~(Alg.~\ref{alg:bitcoin_protocol_strategy}).

As both protocols mint at the same rate, their respective tokens have the same~$\externalCurrency$ value (cf. \S\ref{sec:model_exchange_rate}).
However,~$\protocolNameOne$ miners get less reward ($\oneMinusPosRatioBrackets \rewardPerBlock$ compared to~$\rewardPerBlock$), hence incentivized to spend less electricity in equilibrium (Eq.~\ref{eq:model_total_utilities_equals_expenses}).
As such, it provides~$\protocolPropertyPow = 1-\posRatio$, achieving less wastefulness.

However, as reward for creating a block is now lower also are the costs to create one, and as such, the system is less resilient to free safety-violation attacks.
Specifically, it only provides~$\protocolPropertyDoubleSpendAttack = 1 - \posRatio$ compared to the original~$\protocolPropertyDoubleSpendAttack = 1$.

That is,~$\protocolNameOne$ reduces wastefulness at the expense of making the system less resilient against rational attackers.

\subsection{$\protocolNameMandatoryFullName$}
\label{sec:protocol_v3_strawman}

The motivation for~$\protocolNameMandatory$ is to incorporate internal expenses as part of the block creation process.
It is just like~$\hebShortName$, but requires internal expenditure to create blocks.
That is, there is a single type of blocks, and a miner~$i$ can create at most~$\floor*{\tfrac{\balanceSubSuper{i}{\internalCurrency}}{\posRatio\cdot \rewardPerBlock}}$ blocks in an epoch on chain~$\mainChainName{}$.

The prescribed strategy~$\strategyNameProtocol{\defaultStrategyName}{\protocolNameMandatory}$~(Alg.~\ref{alg:v2_protocol_strategy}) indicates that miner~$i$ should allocate her balance s.t.~$\left\langle \posRatio \balanceSub{i}, \oneMinusPosRatioBrackets \balanceSub{i}\right\rangle$, point her created blocks to~$\lastBlockOnMainChain{\funcLongestChains{}}$, and publish them immediately.
In case of conflicting longest chains,~$\strategyNameProtocol{\defaultStrategyName}{\protocolNameMandatory}$ states that the miner points her next block to either of them picked uniformly-at-random.

\begin{algorithm}[t]
	\SetNoFillComment
	\SetAlgoNoLine 
	\SetAlgoNoEnd 
	\DontPrintSemicolon 
	\caption{$\strategyNameProtocol{i,\defaultStrategyName}{\protocolNameMandatory}$} 
	\label{alg:v2_protocol_strategy} 
	
	\BlankLine
	\Fn{\FAllocateBudgt{\funcArgsAllocateBudgt{i}}}{
		\Return~$\left\langle \posRatio \balanceSub{i}, \oneMinusPosRatioBrackets \balanceSub{i} \right\rangle~$\;
	}
	\BlankLine
	
	\BlankLine
	\Fn{\FGenerateBlock{\funcArgsGenerateBlock{i}}}{
		$\mainChainName{} \gets \text{ uniformly at random from } \funcLongestChains{}$\;
		$\text{pointer} \gets \lastBlockOnMainChain{\mainChainName{}}$\;
		\Return NewBlock(pointer)\;
	}
	\BlankLine
	
	\BlankLine
	\Fn{\FPublish{\funcArgsPublish{i}}}{
		\Return All previously unpublished blocks\;
	}
\end{algorithm}

However,~$\protocolNameMandatory$ has two significant shortcomings.
First, miners must obtain~$\internalCurrency$ prior to creating blocks, failing to do so prevents participation.
Therefore,~$\protocolPropertyPermissiveness=0$, i.e., it is a permissioned protocol.

Moreover, in equilibrium, miners are expected to commit~$\internalCurrency$ sufficient to enable the creation of exactly the number of blocks they expect.
Now, assume a miner becomes absent, either maliciously or unintentionally.
Other miners cannot create new blocks exceeding their quota, and the system halts.

These problems arise as~$\protocolNameMandatory$ conditions block creation on early~$\internalCurrency$ expenditure, justifying~$\hebShortName$'s design of \emph{incentivizing, but not forcing} miners to internally spend.


\section{Practical Considerations}
\label{app:practical_considerations}

\paragraph{Epoch duration and~$\protocolPropertyDecentralization$}
Longer epochs are required to reduce~$\protocolPropertyDecentralization$-- they assure that with high probability each miner gets to create her expected number of blocks, and thus not under-utilize her internally-spent tokens.
However, such longer epochs impose longer payout intervals, and are less appealing to miners.

Payout intervals exist even in existing PoW systems (e.g., 100 blocks maturity period in Bitcoin), but, to allow for shorter epochs while not degrading expected utilization, we suggest letting surplus internal expenses carry out to the following consecutive epoch.
Specifically, instead of redistributing all committed internal expenses after every epoch, the system redistributes only the used tokens, while the others are carried out to the next epoch.
That reduces under-utilization, making shorter epochs feasible.


\paragraph{Circulating supply}

Recall that~$\hebShortName$ works under the assumption that~$\balanceOfMinersAll\left(\epochIndex\right) \ll \balanceOfOthers\left(\epochIndex\right)$, thus requiring a sufficiently large circulating supply.
To achieve this, we suggest bootstrapping with a~\emph{ramp up} period, in which~$\hebShortName$ performs as $\protocolNameBitcoin$, allowing enough internal currency to accumulate.

\paragraph{Internal expenditure mechanism}

The cryptocurrency transaction mechanism can be utilized to let users internally expend their tokens.
A specific implementation could be to let the miners transact their tokens to a null address~\cite{stewart2012proof,p4titan2014slimcoin}.

\paragraph{Currency redistribution}

Recall~$\hebShortName$ redistributes the internally-spent currency among all the cryptocurrency holders, proportionally to their relative holdings.
There are two issues at hand~-- first, how to divide the tokens proportionally (who gets what), and then, how to perform the actual redistribution (associate entities with their tokens).

Blockchain systems record user amounts, whether they are account~\cite{buterin2013ethereum} or UTXO~\cite{nakamoto2008bitcoin,litecoin2013site,dogecoin2013site,sasson2014zerocash} based.
As such, system users know all relative token holdings and therefore how to proportionally divide the required amount.
We note proportional distribution could require fragmenting the atomic unit of the currency (e.g., 1 Satoshi in Bitcoin).
One can circumvent such scenarios by using pseudo-random tie breaking, rewarding only one user with the atomic unit~\cite{pass2015micropayments}.
We illustrate through an example: say there are 10 users with the same token holdings, due to share 15 tokens, that is, each is due 1.5 tokens.
As such, 5 users selected at random will receive 2 tokens, while the other receive only 1.

The actual token allocation can be performed implicitly, that is, without including explicit transactions in the blockchain.


\section{$\hebShortName$ $\protocolPropertyPermissiveness$ Analysis}
\label{app:permissiveness}

We note the miner balance distribution affects miner utilities, and thus affects~$\protocolPropertyPermissiveness$ as well.

We consider a system with two miners~$i$ and~$\neg i$, with balances~$\balanceSubSuper{i}{}, \balanceSubSuper{\neg i}{}$, respectively.
Miner~$\neg i$ has access to~$\internalCurrency$ and follows~$\strategyNameProtocol{\defaultStrategyName}{\hebShortName}$.

If miner~$i$ can obtain~$\internalCurrency$ then she follows~$\strategyNameProtocol{\defaultStrategyName}{\hebShortName}$, resulting with~$\utilityYesIcfuncName{i}{\hebShortName} = \epochLength \balanceNormalizedSub{i}$ (see~\S\ref{sec:protocol_v3_prescribed_strategy_analysis}).

Otherwise, she follows~$\strategyNameProtocol{\noIcStrategyName}{\hebShortName}$, thus allocating her balance~$\balanceSubSuper{i}{\internalCurrency} = 0 ,\balanceSubSuper{i}{\externalCurrency} = \balanceSubSuper{i}{}$.
Consequently,~$\balanceNormalizedSubSuper{i}{\externalCurrency} = \tfrac{\balanceSubSuper{i}{}}{\balanceSubSuper{i}{} + \oneMinusPosRatioBrackets \balanceSubSuper{\neg i}{}}$ and~$\balanceNormalizedSubSuper{\neg i}{\externalCurrency} = \tfrac{\oneMinusPosRatioBrackets \balanceSubSuper{\neg i}{}}{\balanceSubSuper{i}{} + \oneMinusPosRatioBrackets \balanceSubSuper{\neg i}{}}$.

Miner~$i$ creates only regular blocks, with an expected number of~$\tfrac{\balanceSubSuper{i}{}}{\balanceSubSuper{i}{} + \oneMinusPosRatioBrackets \balanceSubSuper{\neg i}{}}$, therefore her expected block weight is~$\E\left[\weight{i}\right] = \epochLength \tfrac{\balanceSubSuper{i}{}}{\balanceSubSuper{i}{} + \oneMinusPosRatioBrackets \balanceSubSuper{\neg i}{}}$.
In contrast, miner~$\neg i$ creates only factored blocks with an expected block weight of~$\E\left[\weight{\neg i}\right] = \epochLength \factor \tfrac{\oneMinusPosRatioBrackets \balanceSubSuper{\neg i}{}}{\balanceSubSuper{i}{} + \oneMinusPosRatioBrackets \balanceSubSuper{\neg i}{}}$. 
Consequently,~$\utilityNoIcfuncName{i}{\hebShortName} = \tfrac{\balanceSubSuper{i}{}}{\balanceSubSuper{i}{} + \factor \balanceSubSuper{\neg i}{}}\epochLength$, 
and as such~$\protocolPropertyPermissiveness = \tfrac{1}{\balanceNormalizedSubSuper{i}{} + \factor \left(1 - \balanceNormalizedSubSuper{i}{}\right)}$.

\end{document}